\documentclass[aps,prx,twocolumn,superscriptaddress,longbibliography,nofootinbib,amsmath,amssymb]{revtex4-2} 
\usepackage[T1]{fontenc}
\usepackage{subcaption} 
\usepackage{physics}
\usepackage{bm}
\usepackage{bbm}
\usepackage[dvipsnames]{xcolor}
\usepackage{hyperref}
\usepackage{dsfont}
\usepackage{mathrsfs}
\usepackage{todonotes}
\usepackage{enumerate}
\usepackage{ragged2e}
\usepackage{amsthm}
\usepackage[normalem]{ulem}
\usepackage{cancel}

\usepackage{tikz}
\usetikzlibrary{calc}
\usetikzlibrary{patterns.meta}
\usetikzlibrary{knots}
\definecolor{mpdocolor}{HTML}{FCDE70}
\definecolor{mpucolor}{HTML}{FCDE70}
\definecolor{mpscolor}{HTML}{E5D9F2}
\definecolor{lcolor}{HTML}{D9EAFD}

\newcommand\mthick{thick}

\newcommand{\MPS}[4]{
	\begin{scope}[shift={(#1)}]
		\draw[fill=#3] (-#2/2,-#2/2) -- (-#2/2,#2/2) -- (#2/2,#2/2) -- (#2/2,-#2/2) -- cycle;
		\draw (0,0) node {#4};
	\end{scope}
}

\newcommand{\MPU}[4]{
	\begin{scope}[shift={(#1)}]
		\draw[fill=#3] (0,0) circle (#2/2);
		\draw (0,0) node {#4};
	\end{scope}
}

\newcommand{\transfermat}[5]{
	\begin{scope}[shift={(#1)}]
		\draw[fill=#4] (-#3/2,-#2/2) -- (-#3/2,#2/2) -- (#3/2,#2/2) -- (#3/2,-#2/2) -- cycle;
		\draw (0,0) node {#5};
	\end{scope}
}

\newcommand{\GcTensor}[6]{
    \begin{scope}[shift={(#1)}]
    \ifnum#5=0
		\draw[thick] (-#2,0) -- (#2,0);
		\draw[thick] (0,#2) -- (0,-#2);
    \fi
    \ifnum#5=-1
		\draw[thick] (0,0) -- (#2,0);
		\draw[thick] (0,#2) -- (0,-#2);
    \fi
    \ifnum#5=1
		\draw[thick] (-#2,0) -- (0,0);
		\draw[thick] (0,#2) -- (0,-#2);
    \fi

    \ifnum#5=2
		\draw[thick] (-#2,0) -- (#2,0);
		\draw[thick] (0,0) -- (0,#2);
    \fi
    \ifnum#5=-2
		\draw[thick] (-#2,0) -- (#2,0);
		\draw[thick] (0,0) -- (0,-#2);
    \fi

    \ifnum#5=3
		\draw[thick] (-#2,0) -- (#2,0);
		\draw[thick] (0,-#2) -- (0,#2);
    \fi
        \draw[thick, fill=#6, rounded corners=2pt] (-#3,-#3) rectangle (#3,#3);
		\draw (0,0) node {\scriptsize #4};
	\end{scope}
}

\newcommand{\projector}[6]{
\begin{scope}[shift={(#1)}]
	\ifnum#2 = 1
		\draw[\mthick](-#3,#4) -- (0, #4);
		\draw[\mthick](-#3,-#4) -- (0,-#4);
		\draw[\mthick](0,0) -- (#3,0);
	\fi
	\ifnum#2 = -1
		\draw[\mthick](#3,#4) -- (0, #4);
		\draw[\mthick](#3,-#4) -- (0,-#4);
		\draw[\mthick](0,0) -- (-#3,0);
	\fi
	\draw[ thick, fill=#6, rounded corners=2pt] (-#3*0.5,-#4*1.2) rectangle (#3*0.5,#4*1.2);
		\draw (0,0) node {\scriptsize #5};
\end{scope}
}

\makeatletter 
\renewcommand\onecolumngrid{
\do@columngrid{one}{\@ne}%
\def\set@footnotewidth{\onecolumngrid}
\def\footnoterule{\kern-6pt\hrule width 1.5in\kern6pt}%
}

\renewcommand\twocolumngrid{
        \def\footnoterule{
        \dimen@\skip\footins\divide\dimen@\thr@@
        \kern-\dimen@\hrule width.5in\kern\dimen@}
        \do@columngrid{mlt}{\tw@}
}%

\def\amsbb{\use@mathgroup \M@U \symAMSb}
\makeatother

\newtheorem{theorem}{Theorem}[section]
\newtheorem{lemma}[theorem]{Lemma}
\newtheorem{proposition}[theorem]{Proposition}
\newtheorem{corollary}{Corollary}[section]
\newtheorem{definition}{Definition}

\newcommand{\ie}{\emph{i.e.\@} }

\newcommand{\id}{\mathds{1}}

\DeclareMathOperator{\arctanh}{arctanh}

\def\cell{0.6} 

\usepackage{xcolor}
\definecolor{yuhan}{rgb}{0.9, 0, 0.5}

\begin{document}
\title{Structure of matrix product locally purifiable density operators}
\author{Yale Yauk}
\email{yale.yauk@mpq.mpg.de}
\affiliation{Max-Planck-Institut f\"ur Quantenoptik, 85748 Garching, Germany}
\affiliation{Munich Center for Quantum Science and Technology (MCQST), 80799 M\"unchen, Germany}
\author{Yuhan Liu}
\affiliation{Max-Planck-Institut f\"ur Quantenoptik, 85748 Garching, Germany}
\affiliation{Munich Center for Quantum Science and Technology (MCQST), 80799 M\"unchen, Germany}
\author{Ignacio Cirac}
\affiliation{Max-Planck-Institut f\"ur Quantenoptik, 85748 Garching, Germany}
\affiliation{Munich Center for Quantum Science and Technology (MCQST), 80799 M\"unchen, Germany}

\begin{abstract}
Tensor network methods provide powerful analytical and numerical tools for characterizing quantum phases of matter. While the mathematical structure of matrix product states (MPS) is well understood through the MPS fundamental theorem, an analogous understanding for mixed-state tensor networks remains largely absent: if two purification tensors generate the same density matrix, how are they related? In this work, we initiate the study of a fundamental theorem for matrix product locally purifiable density operators (LPDOs) and focus on sequentially generated LPDOs (sLPDOs), a broad subclass admitting an interpretation in terms of successive applications of quantum channels on an initial state. We prove that, under suitable invertibility or cyclic conditions, two sLPDO representations generate the same density matrix for arbitrary system sizes if and only if they are related by a matrix product isometry acting on the purification bonds. Beyond the sLPDO setting, we provide a counterexample that suggests an obstruction to a general fundamental theorem for LPDOs with periodic boundary conditions. Finally, we discuss implications for mixed-state symmetry-protected topological phases, including the possibility of nontrivial phases protected only by weak symmetry conditions.
\end{abstract}

\maketitle

\section{Introduction}
Finding and characterizing new quantum phases of matter is a central focus in modern physics. One way to approach this monumental task is to restrict attention to a set of physically relevant states, commonly those arising as ground states of short-range interacting Hamiltonians. Locality imposes strong constraints on the entanglement structure of such states, enabling powerful analytical and numerical techniques known as tensor network methods \cite{TN_review,DMRG_TN}. As a paradigmatic example, matrix product states (MPS) in one-dimensional systems efficiently represent ground states of gapped local Hamiltonians, achieving notable success in the study of pure-state phases \cite{SPT_gapped,SPT_TN}. More recently, advances in experimental quantum simulation and control have extended these questions to mixed states and finite-temperature settings, where a richer understanding of quantum phases is needed. One can approach this question using the mixed-state analog of the MPS, the matrix product density operator (MPDO) \cite{MPDO,MPDO_dynamics}, which provides an efficient description of thermal states of gapped local Hamiltonians \cite{local_gapped_FF,PEPO_Gibbs,MPDO_Gibbs}, as well as boundary states of two-dimensional topologically-ordered systems \cite{PEPS_entanglement_Hamiltonian,PEPS_TO_transfer_matrix,boundary_MPO}.

An important subclass of MPDOs is given by locally purifiable density operators (LPDOs), which naturally arise as states generated by low-depth noisy quantum circuits. The additional structure provides analytical tractability -- for instance, in the study of mixed-state symmetry-protected topological (SPT) phases \cite{LPDO_SPT} -- and numerical advantages, as the dynamics preserve the density matrix structure.

Although LPDOs have proven empirically useful, their mathematical structure is not well-understood. For example, two LPDO representations can be equivalent, in the sense that two different local tensors generate the same global density matrix for arbitrary system sizes. Characterizing the freedom within equivalent tensors is the central problem addressed in this work, with the answer formulated as a fundamental theorem. Such a theorem is significant for two reasons. On the analytical side, it would help to lead to the complete characterization of mixed-state SPT phases, in much the same way that the MPS fundamental theorem \cite{MPS_rep} underpins the full classification of 1D SPT phases~\cite{SPT_gapped,SPT_TN}. On the numerical side, understanding the equivalence of LPDO representations can help identify and eliminate redundant parameters, leading to more efficient algorithms. 

The degrees of freedom associated with equivalent LPDO representations are expected to be substantially richer than those appearing in MPS. In particular, the freedom may reside not only in the virtual bonds, but also in the purification bonds. For example, two LPDO representations related by a matrix product unitary (MPU) acting on the purification bonds necessarily generate the same mixed state. Proving or disproving the converse statement, that this is the entire freedom, is challenging. 

In this paper, we take a first step toward establishing a fundamental theorem for LPDOs. In particular, we focus on a subclass of LPDOs, the sequentially generated LPDOs (sLPDOs). This class is highly expressive and captures many density matrices of physical interest. Crucially, sLPDOs admit an interpretation in terms of successive applications of quantum channels acting on an initial state, which allows us to connect their structure to properties of quantum channels. Within the sLPDO framework, we prove two sufficient conditions on the tensors such that a fundamental theorem holds: that two purification tensors generate the same density matrices for all system sizes if and only if they are related by a matrix product isometry acting on the purification bonds and conjugation on the resulting virtual bonds. We further present a counterexample beyond the sLPDO setting. Specifically, we construct two LPDO representations with periodic boundary conditions that generate the same density matrices, and show that they cannot be connected by any matrix product unitary of bond dimension independent of system size. This counterexample suggests that a fundamental theorem for general LPDOs with periodic boundary conditions can only exist under further assumptions, such as a reduction to canonical form, as in the MPS case. A canonical form in this context would necessarily intertwine the virtual and purification degrees of freedom. We conclude by discussing the implications of the LPDO fundamental theorem for mixed-state SPT phases, showing in particular that a density matrix can exhibit a nontrivial phase even in the presence of only weak symmetry conditions.

We outline the structure of the paper. In section~\ref{sec:prelim}, we introduce the tensor network notation and well-known results about purifications which will be the basis of our study. Section~\ref{sec:Fundamental_theorem} presents the main results -- two generic sufficient conditions for a fundamental theorem connecting two sLPDO tensors generating the same state. In section~\ref{sec:examples}, we show that the sLPDO ansatz captures important classes of mixed states, in particular, all boundaries of $D(G)$ topological order. The rest of the paper deals with tensors beyond the sLPDO ansatz. In section~\ref{sec:counterexample}, we present an example of two tensors generating the same state, whose ancilla spaces are not related by a local MPU: there cannot exist local degrees of freedom relating these two tensors. Lastly, in section~\ref{sec:SPT}, we identify a density operator with a seemingly trivial on-site symmetry, though could still be in a non-trivial phase, going beyond the classifications elucidated in previous literature. We finish with a discussion on obstacles for a full fundamental theorem for LPDOs.

\section{Preliminaries and notation}
\label{sec:prelim}
In this work, we study representations of many-body mixed states and their equivalence, that is, how different representations can generate the same state. For this purpose, we must also consider their purifications, which are pure states on a larger combined system-ancilla Hilbert space such that the mixed state can be obtained by tracing out the ancilla degrees of freedom. We will concern ourselves only with those mixed states whose purifications admit an exact tensor network representation.

Given a density matrix $\rho$ on $\mathbb{C}^d$, a purification is a pure state $\ket{\Psi}\in \mathbb{C}^d\otimes \mathbb{C}^r$ such that $\rho=\tr_{\mathbb{C}^r} \ketbra{\Psi}{\Psi}$. Canonically, $L(\mathbb{C}^r,\mathbb{C}^d)\cong \mathbb{C}^d\otimes (\mathbb{C}^r)^*$, and after choosing a basis, we identify $(\mathbb{C}^r)^*$ with $\mathbb{C}^r$. Thus $\ket{\Psi}$ corresponds to a matrix $\sigma:\mathbb{C}^r\to\mathbb{C}^d$ satisfying $\rho=\sigma\sigma^\dagger$. By abuse of terminology, we call both $\ket{\Psi}$ and $\sigma$ purifications of $\rho$. We say that a purification is minimal if there does not exist another purification with a smaller ancillary dimension, equivalently, that $\sigma$ is full-rank.

We begin by introducing a key result for purifications that will be frequently referred to in this work. Then, we introduce the tensor network formalism along with the different families of tensor networks that will appear later. We end the section by clearly outlining the problem studied in this work.

It is understood throughout that we work in the field of complex numbers and we denote the set of complex $m\times n$ matrices by $\mathcal{M}_{m,n}$ and by $\mathcal{M}_m$ when $m=n$. 

\subsection{Freedom in purifications}
A basic result in quantum information theory says that any two purifications of a given state are related by a partial isometry. Because this plays a central role in the main result, we will restate the fact carefully in full.

\begin{lemma}
\label{lem:purification_partial_isometry}
    Let $A\in\mathcal{M}_{m,p_A}$ and $B\in\mathcal{M}_{m,p_B}$. Then
    \begin{equation}
        AA^\dagger=BB^\dagger
    \end{equation} 
    if and only if there exists a partial isometry $U_0:\mathbb{C}^{p_A}\to\mathbb{C}^{p_B}$ such that
    \begin{equation}
    \label{eq:partial_isometry_properties}
        A=BU_0\,, \quad U_0^\dagger U_0=\mathrm{Proj}_{\mathrm{ker}(A)^\perp}\,, \quad U_0U^\dagger_0=\mathrm{Proj}_{\mathrm{ker}(B)^\perp}\,.
    \end{equation}
    The partial isometry $U_0$ is unique if it exists. Furthermore, $U_0$ ($U^\dagger_0$) is an isometry if $A$ ($B$) has full column rank. 
\end{lemma}
A proof is provided in appendix~\ref{app:proof}. 

We emphasize that this lemma is an existence statement, which picks a canonical unique solution to $A=BU$. The full freedom is given by $U=U_0+K$, where $K$ is an arbitrary map satisfying $\mathrm{col}(K)\subseteq\mathrm{ker}(B)$. 

\begin{corollary}
\label{cor:purification_freedom}
    Let $\rho$ be a density matrix on $\mathbb{C}^d$ of rank $r$. Let $\sigma:\mathbb{C}^r\to \mathbb{C}^d$ be a minimal purification of $\rho$. For any other purification $\eta:\mathbb{C}^s\to \mathbb{C}^d$ of $\rho$ with $s\geq r$, there exists a unique isometry $V:\mathbb{C}^s\to\mathbb{C}^r$ with $V V^\dagger=\id_{\mathbb{C}^{r}}$ and $V^\dagger V=\mathrm{Proj}_{\mathrm{ker}(\eta)^\perp}$ such that
    \begin{equation}
        \eta=\sigma V\,.
    \end{equation}
\end{corollary}
To be precise, $V$ is known as a coisometry (that $V^\dagger$ is an isometry), which is a reflection of the notation convention. Rather, it should be just interpreted as an isometric embedding of the minimal ancilla space of $\sigma$ into that of $\rho$.

\subsection{Tensor networks}
A tensor network is a quantum state or operator built out of local rank-$k$ tensors $A_{[n]}$, which in this context is just a multilinear map. They admit a graphical calculus which will be used throughout. As an example, we will repeatedly encounter the rank-4 tensor $A_{[n]}$ 
\begin{equation}
    (A_{[n]})^{ia}_{\alpha\beta}=\begin{tikzpicture}[baseline={([yshift=-0.5ex]current bounding box.center)},x=\cell cm,y=\cell cm]
            \draw (-1,0) node [left] {\scriptsize $\alpha$} -- (1,0) node [right] {\scriptsize $\beta$};
            \draw[dashed] (0,-1) node [below] {\scriptsize $a$} -- (0,0);
            \draw (0,0) -- (0,1) node [above] {\scriptsize $i$};
            \MPS{0,0}{1}{blue!20}{\scriptsize $A_{[n]}$} 
        \end{tikzpicture}\in\mathbb{C}\,,
\end{equation}
with indices $i=1,\dots,d$ corresponding to the physical dimension, $a=1,\dots,p$ to an ancillary dimension and $\alpha,\beta=1,\dots,D_{n-1/n}$ corresponding to the left/right virtual (bond) dimension. We often omit the virtual indices and understand $A^{ia}_{[n]}\in\mathcal{M}_{D_{n+1},D_n}$ as a $D_{n+1}\times D_n$-dimensional matrix.

We will encounter various variations of tensor networks, so we present them under a generalized framework, as matrix product vectors (MPVs). Let $\{A_{[n]}\}_{n=1}^N$ be a set of rank-3 tensors.
\begin{definition}
    An MPV is a vector of the form
    \begin{align}
        \ket{\Psi^{(N)}(\{A_{[n]}\})}&=\sum_{i_N,\dots,i_N} \tr(A^{i_1}_{[N]}\cdots A^{i_2}_{[2]} A^{i_1}_{[1]})\ket{i_N,\dots,i_1} \nonumber \\
        &=\begin{tikzpicture}[baseline={([yshift=-0.5ex]current bounding box.center)},x=\cell cm,y=\cell cm]
        \draw (-1,0) -- (1,0);
        \node at (1.5,0) {$\cdots$};
        \draw (2,0) -- (5.5,0);
        \draw (0,0) -- (0,1);
        \draw (3,0) -- (3,1);
        \draw (4.5,0) -- (4.5,1);
        \MPS{0,0}{1}{blue!20}{\scriptsize $A_{[N]}$}
        \MPS{3,0}{1}{blue!20}{\scriptsize $A_{[2]}$}
        \MPS{4.5,0}{1}{blue!20}{\scriptsize $A_{[1]}$}
        \draw (-1,0) to [out=180,in=180] (-1,-0.5);
        \draw (5.5,-0.5) to [out=0,in=0] (5.5,0);
    \end{tikzpicture}\,.
    \end{align}
\end{definition}
The state has periodic boundary conditions (PBC) and if $A^i_{[n]}=A^i$ for all $1\leq n\leq N$, we say that the MPV is homogeneous and write $\ket{\Psi^{(N)}(A)}$.

The definition of MPVs includes the usual MPSs and also matrix product operators (MPOs) by treating the physical dimension as an index on a vectorized Hilbert space. A subclass that we will encounter throughout is the class of tensor network unitaries \cite{MPU,MPU_nontI}.
\begin{definition}
    A matrix product unitary (MPU) is an MPO that generates a unitary operator for all system sizes $N\geq 1$.
\end{definition}
\noindent More generally, a matrix product isometry (MPI) is an MPO that generates an isometry for all system sizes $N\geq 1$.

Finally, we define the object of focus in the paper, the MPDO.
\begin{definition}
    A PBC MPDO is a positive-semidefinite MPO
    \begin{align}
    \label{eq:MPDO}
        \rho&^{(N)}(\{M_{[n]}\})\nonumber
        \\&=\hspace{-2mm}\sum_{\substack{i_1,\dots,i_N \\ j_1,\dots,j_N}}\hspace{-2mm}\tr(M^{i_Nj_N}_{[N]}\cdots M^{i_1j_1}_{[1]})\ketbra{i_N,\dots,i_1}{j_N,\dots,j_1} \nonumber \\
        &=\begin{tikzpicture}[baseline={([yshift=-0.5ex]current bounding box.center)},x=\cell cm,y=\cell cm]
        \draw (-1,0) -- (1,0);
        \node at (1.5,0) {$\cdots$};
        \draw (2,0) -- (5.5,0);
        \draw (0,-1) -- (0,1);
        \draw (3,-1) -- (3,1);
        \draw (4.5,-1) -- (4.5,1);
        \MPS{0,0}{1}{red!20}{\scriptsize $M_{[N]}$}
        \MPS{3,0}{1}{red!20}{\scriptsize $M_{[2]}$}
        \MPS{4.5,0}{1}{red!20}{\scriptsize $M_{[1]}$}
        \draw (-1,0) to [out=180,in=180] (-1,-0.5);
        \draw (5.5,-0.5) to [out=0,in=0] (5.5,0);
    \end{tikzpicture}\,.
\end{align}
\end{definition}

The problem of determining whether a given MPO is an MPDO, that is, that the operator is positive and has unit trace, has been shown to be undecidable in the thermodynamic limit and NP-hard in finite systems \cite{MPDO_NPhard}. To make progress, it is common to consider subclasses of MPDOs called locally purifiable density operators (LPDOs), which are those density operators for which there exists an MPS purification and are positive semi-definite by construction. However, it is also known that LPDOs cannot represent the whole class, because not all MPDOs can be written as LPDOs \cite{MPDO_LPDO}. Nevertheless, we will work within this ansatz since it includes many mixed states of interest, for example, boundaries of two-dimensional topologically-ordered states~\cite{PEPS_entanglement_Hamiltonian,PEPS_TO_transfer_matrix,MPDO_RFP,boundary_MPO} and a class of thermal states of local Hamiltonians~\cite{MPDO_Gibbs}. The precise definition of LPDO is as follows~\cite{MPDO}.

\begin{definition}
    An LPDO is an MPDO with the further assumption that each local tensor $M_{[n]}$ can be locally purified, \ie there exists a rank-4 tensor $A_{[n]}$ such that
    \begin{equation}
    \label{eq:LPDO}
        M^{ij}_{[n]}=\sum_{a=1}^p A^{ia}_{[n]}\otimes \overline{A}{}^{ja}_{[n]}\,.
    \end{equation}
    \end{definition}
\noindent Graphically,
    \begin{equation}
        \begin{tikzpicture}[baseline={([yshift=-0.5ex]current bounding box.center)},x=\cell cm,y=\cell cm]
            \draw[line width=1.2pt] (-1,0) -- (1,0);
            \draw (0,-1) -- (0,1);
            \MPS{0,0}{1}{red!20}{\scriptsize $M_{[n]}$}
        \end{tikzpicture}=\begin{tikzpicture}[baseline={([yshift=-0.5ex]current bounding box.center)},x=\cell cm,y=\cell cm]
            \draw (-1,0) -- (1,0);
            \draw (0,-1) -- (0,0);
            \draw[dashed] (0,0) -- (0,1.5);
            \draw (0,1.5) -- (0,2.5);
            \draw (-1,1.5) -- (1,1.5);
            \MPS{0,0}{1}{blue!20}{\scriptsize $\overline{A}_{[n]}$}
            \MPS{0,1.5}{1}{blue!20}{\scriptsize $A_{[n]}$}
        \end{tikzpicture}\,.
    \end{equation}
We refer to the rank-4 tensor $A_{[n]}$ as the purification tensor of the LPDO. An equivalent way to see the state is that $\{A_{[n]}\}$ generates an MPO
\begin{equation}
\begin{aligned}
    \sigma&^{(N)}(\{A_{[n]}\})\\
    &=\hspace{-2mm}\sum_{\substack{i_1,\dots,i_N \\ a_1,\dots,a_N}}\hspace{-2mm}\tr(A^{i_Na_N}_{[N]}\cdots A^{i_1a_1}_{[1]})\ketbra{i_N,\dots,i_1}{a_N,\dots,a_1}\,,
\end{aligned}
\end{equation}
which, when vectorized, corresponds to a purification of the LPDO $\rho^{(N)}(\{M_{[n]}\})$. Thus $\rho^{(N)}(\{M_{[n]}\})=\sigma^{(N)}(\{A_{[n]}\})\sigma^{(N)}(\{A_{[n]}\})^\dagger$.

\subsection{Sequentially generated LPDOs}
\label{sec:def-sLPDO}
We introduce a class of LPDOs which we call sequentially generated LPDOs (sLPDOs). These coincide with finite-system, single-boundary finitely-correlated states constructed in \cite{FNW}. Given any purification tensor $A^{ia}_{\alpha\beta}$, we can associate to it a 1-input/2-output completely positive map $\mathcal{E}_A(\rho)=\sum_{i,j}\left(\sum_{a}A^{ia}\rho(A^{ja})^\dagger\right)\otimes \ketbra{i}{j}$ by viewing $A^a:\mathbb{C}^D\to\mathbb{C}^D\otimes\mathbb{C}^d$ as its Kraus operators. We define a state that is generated sequentially by applying $\mathcal{E}_A$ on the virtual dimension at each step, starting with an initial mixed state $\omega$ on $\mathbb{C}^D$. We interpret the final virtual leg as a physical index, giving a state on $N+1$ sites.

\begin{definition}
    An sLPDO is a state of the form
    \begin{align}
    \label{eq:sequential_LPDO}
    \rho^{(N)}(\{A_{[n]}\};\omega)&=(\mathcal{E}_{A_{[N]}}\otimes \mathrm{id}_{(\mathbb{C}^d)^{\otimes N-1}})\circ\cdots \nonumber \\
    &\phantom{={}}\circ(\mathcal{E}_{A_{[2]}}\otimes \mathrm{id}_{\mathbb{C}^d})\circ\mathcal{E}_{A_{[1]}}(\omega) \nonumber \\
        &= \begin{tikzpicture}[baseline={([yshift=-0.5ex]current bounding box.center)},x=\cell cm,y=\cell cm]
            \draw (1,0) -- (-1,0) -- (-1,1);
            \node at (1.5,0) {$\cdots$};
            \draw (2,0) -- (5.5,0) -- (5.5,-1);
            \draw (0,0) -- (0,1);
            \draw[dashed] (0,-1) -- (0,0);
            \draw (4.5,0) -- (4.5,1);
            \draw[dashed] (4.5,-1) -- (4.5,0);
            \draw (3,0) -- (3,1);
            \draw[dashed] (3,-1) -- (3,0);
            \MPS{0,0}{1}{blue!20}{\scriptsize $A_{[N]}$}
            \MPS{3,0}{1}{blue!20}{\scriptsize $A_{[2]}$}
            \MPS{4.5,0}{1}{blue!20}{\scriptsize $A_{[1]}$}
            \begin{scope}[shift={(0,-0.75)}, yscale=-1, shift={(0,0.75)}]
            \draw (1,0) -- (-1,0) -- (-1,1);
            \node at (1.5,0) {$\cdots$};
            \draw (2,0) -- (5.5,0) -- (5.5,-1);
            \draw (0,0) -- (0,1);
            \draw (4.5,0) -- (4.5,1);
            \draw (3,0) -- (3,1);
            \MPS{0,0}{1}{blue!20}{\scriptsize $A_{[N]}$}
            \MPS{3,0}{1}{blue!20}{\scriptsize $A_{[2]}$}
            \MPS{4.5,0}{1}{blue!20}{\scriptsize $A_{[1]}$}
            \end{scope}
            \MPU{5.5,-0.75}{0.75}{green!20}{\scriptsize $\omega$}
        \end{tikzpicture}\,.
    \end{align}
    We say that the sLPDO is generated by the pair $(\{A_{[n]}\},\omega)$. If all tensors $A_{[n]}=A$ are site independent, it is a homogeneous sLPDO and we write $\rho^{(N)}(A;\omega)$.
\end{definition}
We define $\sigma^{(N)}(\{A_{[n]}\};\sqrt{\omega})$ as the purification of the sLPDO generated by $(\{A_{[n]}\},\omega)$, that is,
\begin{equation}
    \sigma^{(N)}(\{A_{[n]}\};\sqrt{\omega})=\begin{tikzpicture}[baseline={([yshift=-0.5ex]current bounding box.center)},x=\cell cm,y=\cell cm]
            \draw (1,0) -- (-1,0) -- (-1,1);
            \node at (1.5,0) {$\cdots$};
            \draw (2,0) -- (5.75,0) -- (5.75,-1);
            \draw (0,0) -- (0,1);
            \draw[dashed] (0,-1) -- (0,0);
            \draw (4.5,0) -- (4.5,1);
            \draw[dashed] (4.5,-1) -- (4.5,0);
            \draw (3,0) -- (3,1);
            \draw[dashed] (3,-1) -- (3,0);
            \MPS{0,0}{1}{blue!20}{\scriptsize $A_{[N]}$}
            \MPS{3,0}{1}{blue!20}{\scriptsize $A_{[2]}$}
            \MPS{4.5,0}{1}{blue!20}{\scriptsize $A_{[1]}$}
            \MPU{5.75,0}{0.75}{green!20}{\scriptsize $\sqrt{\omega}$}           
        \end{tikzpicture}\,,
\end{equation}
where $\sqrt{\omega}\sqrt{\omega}{}^\dagger=\omega$. Here, we abuse the notation $\sqrt{\omega}$, because we could generally allow any non-square full-rank matrix satisfying $\sqrt{\omega}\sqrt{\omega}{}^\dagger=\omega$. If $\omega$ has rank $r$, then $\sigma^{(N)}(\{A_{[n]}\};\sqrt{\omega})$ can always be represented as an operator $\sigma^{(N)}(\{A_{[n]}\};\sqrt{\omega}):(\mathbb{C}^p)^{\otimes N}\otimes \mathbb{C}^r\to\mathbb{C}^D\otimes (\mathbb{C}^d)^{\otimes N}$.

We note that for a general tensor $A$, the associated completely-positive map $\mathcal{E}_A$ need not be trace-preserving. There exists a gauge transformation $XA^{ia}X^{-1}$ to locally make it so -- however while on internal bonds the gauges cancel, the final dangling memory leg retains a gauge matrix $X$, so is generally not the same sLPDO.

We recall the motivation for studying sequentially generated states in the case of MPS. Protocols to prepare such multiqudit states involve a Markovian open quantum system modeled by an emitter Hilbert space $\mathcal{H}_{\text{em}}\cong \mathbb{C}^D$, sequentially leaking information in the form of qudits into the environment $\bigotimes_{n=1}^N \mathcal{H}_n\cong (\mathbb{C}^d)^{\otimes N}$ \cite{MPS_sequential}. This successive leaking is modeled by an isometric interaction $V_{[n]}:\mathcal{H}_{\text{em}}\to\mathcal{H}_{\text{em}}\otimes\mathcal{H}_n$ acting on an initial emitter state $\ket{R}\in \mathcal{H}_{\text{em}}$. Conditioned on a final emitter state $\ket{L}\in\mathcal{H}_{\text{em}}$, the multiqudit state is an MPS generated by tensors $\{V_{[n]}\}_{n=1}^N$ and boundary vectors $\ket{L}$ and $\ket{R}$.

sLPDOs are mixed state generalizations of sequentially generated states, and can be seen in two equivalent ways by Stinespring dilation. Either, that the leaking interaction is replaced by a quantum operation $\mathcal{E}_A$, or that parts of the final pure multiqudit state are traced out. The only other difference is that we retain the emitter system such that the final mixed state lives on $\mathcal{H}_{\text{em}}\otimes\bigotimes_{n=1}^N \mathcal{H}_n$.

Another motivation for considering sLPDOs stems from their physical relevance with Gibbs states. In particular, when $\mathcal{E}_A$ is a unital quantum channel, the resulting $\rho^{(N)}(A;\omega)$ exhibits decaying conditional mutual information and hence can be represented as a Gibbs state of a local Hamiltonian~\cite{MPDO_Gibbs}.

As it turns out, any PBC LPDO admits an sLPDO representation at the cost of squaring the bond dimension. Starting from a state in the uniform LPDO representation generated by purification tensors $\{\tilde{A}_{[n]}\}$, we will find its sLPDO representation $\{A_{[n]}\}$.

Consider the PBC LPDO purification $\sigma^{(N)}(\{\tilde{A}_{[n]}\})$. Viewing $\tilde{A}_{[N]}$ as a map $\tilde{A}_{[N]}:\mathbb{C}^{d}\otimes \mathbb{C}^{D_N}\to \mathbb{C}^p\otimes \mathbb{C}^{D_{N+1}}$, perform any matrix factorization, such as a singular value decomposition or a QR decomposition:
\begin{equation}
    \begin{tikzpicture}[baseline={([yshift=-0.5ex]current bounding box.center)},x=\cell cm,y=\cell cm]
            \draw[dashed] (0,0.5) -- (0,1.5);
            \draw (0,1.5) -- (0,2.5);
            \draw (-1,1.5) -- (1,1.5);
            \MPS{0,1.5}{1}{blue!20}{\scriptsize $\tilde{A}_{[N]}$}
        \end{tikzpicture}=\begin{tikzpicture}[baseline={([yshift=-0.5ex]current bounding box.center)},x=\cell cm,y=\cell cm]
            \draw[dashed] (0,-1.0) -- (0,0.);
            \draw (-1,0.) -- (0,0);
            \draw[thick] (0,0) -- (1,1);
            \draw (1,1) -- (2,1);
            \draw (1,1) -- (1,2);
            \MPU{0,0}{0.75}{JungleGreen!40}{}
            \MPU{1,1.}{0.75}{Plum!40}{}
        \end{tikzpicture}.
\end{equation}
Next, we substitute this into the PBC LPDO ansatz and perform a shift on the purification indices, to obtain
\begin{equation}
    \begin{tikzpicture}[baseline={([yshift=-0.5ex]current bounding box.center)},x=\cell cm,y=\cell cm]
        \draw (-1,0) -- (2.5,0);
        \node at (3,0) {$\cdots$};
        \draw (3.5,0) -- (5.5,0);
        \draw (0,0) -- (0,1);
        \draw (1.5,0) -- (1.5,1);
        \draw (4.5,0) -- (4.5,1);
        \draw[dashed] (0,-0.5) .. controls +(0,-1) and +(0,1) .. (4.5,-2);
        \draw[dashed] (1.5,-0.5) .. controls +(0,-1) and +(0,1) .. (0,-2);
        \draw[dashed] (4.5,-0.5) .. controls +(0,-1) and +(0,1) .. (3,-2);
        \MPS{0,0}{1}{blue!20}{\scriptsize $\tilde{A}_{[N]}$}
        \MPS{1.5,0}{1}{blue!20}{\scriptsize $\tilde{A}_{[-]}$}
        \MPS{4.5,0}{1}{blue!20}{\scriptsize $\tilde{A}_{[1]}$}
        \draw (-1,0) to [out=180,in=180] (-1,-0.5);
        \draw (5.5,-0.5) to [out=0,in=0] (5.5,0);
    \end{tikzpicture}=\begin{tikzpicture}[baseline={([yshift=-0.5ex]current bounding box.center)},x=\cell cm,y=\cell cm]
        \draw (-1,0) -- (1.0,0);
        \node at (1.5,0) {$\cdots$};
        \draw (2,0) -- (4.5,0);
        \draw[dashed] (0,-1) -- (0,0);   \draw (0,0) -- (0,1);
        \draw[dashed] (3.0,-1) -- (3.0,0); \draw (3.0,0) -- (3.0,1);
        \draw (-1.2,0) -- (-1.2,1);
        \draw[dashed] (4.2,-1) -- (4.2,0);
        \draw[thick] (4.2,0) to [out=135,in=0] (3.5,-0.75);
        \draw[thick] (3.5,-0.75) -- (-0.5, -0.75);
        \draw[thick] (-1.2,0) to [out=45,in=180] (-0.5, -0.75);
        \MPS{0,0}{1}{blue!20}{\scriptsize $\tilde{A}_{[-]}$}
        \MPS{3.0,0}{1}{blue!20}{\scriptsize $\tilde{A}_{[1]}$}
        \MPU{4.2,0}{0.75}{JungleGreen!40}{}
        \MPU{-1.2,0}{0.75}{Plum!40}{}
    \end{tikzpicture}\,,
\end{equation}
where we abbreviated $A_{[-]}=A_{[N-1]}$ to make it fit. Indeed, this now has the form of an sLPDO, where we identify an sLPDO $\sigma^{(N-1)}(\{A_{[n]}\};\sqrt{\omega})$ with bulk tensors
\begin{equation}
    \begin{tikzpicture}[baseline={([yshift=-0.5ex]current bounding box.center)},x=\cell cm,y=\cell cm]
        \draw[dashed] (0,-1) -- (0,0);
        \draw (0,0) -- (0,1);
        \draw (-1,0) -- (1,0);
        \MPS{0,0}{1}{blue!20}{\scriptsize $A_{[n]}$}
    \end{tikzpicture}=\begin{tikzpicture}[baseline={([yshift=-0.5ex]current bounding box.center)},x=\cell cm,y=\cell cm]
        \draw[dashed] (0,-1) -- (0,0);
        \draw (0,0) -- (0,1);
        \draw (-1,0) -- (1,0);
        \draw[thick] (-1,-0.75)-- (1,-0.75);
        \MPS{0,0}{1}{blue!20}{\scriptsize $\tilde{A}_{[n]}$}
    \end{tikzpicture}\quad \mathrm{for} \ n=1,\dots, N-2\,,
\end{equation}
and boundaries 
\begin{equation}
        \begin{tikzpicture}[baseline={([yshift=-0.5ex]current bounding box.center)},x=\cell cm,y=\cell cm]
            \draw[dashed] (0,0.5) -- (0,1.5);
            \draw (0,1.5) -- (0,2.5);
            \draw (-1,1.5) -- (1,1.5);
            \draw (-1,1.5) -- (-1,2.5);
            \MPS{0,1.5}{1}{blue!20}{\scriptsize $A_{[-]}$}
        \end{tikzpicture}=\begin{tikzpicture}[baseline={([yshift=-0.5ex]current bounding box.center)},x=\cell cm,y=\cell cm]
            \draw (-1,0) -- (1.0,0);
            \draw[thick] (-1.2,0) to [out=45,in=180] (-0.5, -0.75);
            \draw[thick] (1.0,-0.75) -- (-0.5, -0.75);
            \draw (-1.2,0) -- (-1.2,1);
            \draw (0,0) -- (0,1);
            \draw[dashed] (0,-1) -- (0,0);
            \MPS{0,0}{1}{blue!20}{\scriptsize $\tilde{A}_{[-]}$}
            \MPU{-1.2,0}{0.75}{Plum!40}{}
        \end{tikzpicture}\,,\qquad  \begin{tikzpicture}[baseline={([yshift=-0.5ex]current bounding box.center)},x=\cell cm,y=\cell cm]
            \draw (-0.75,-1.25) -- (0,-1.25) -- (0,-2);
            \MPU{0,-1.25}{0.75}{green!20}{\scriptsize $\sqrt{\omega}$}
        \end{tikzpicture}=\begin{tikzpicture}[baseline={([yshift=-0.5ex]current bounding box.center)},x=\cell cm,y=\cell cm]
            \draw (-0.75,-1.25) -- (0,-1.25);
            \draw[dashed] (0,-2) -- (0,-1.25);
            \draw[thick] (0,-1.25) to [out=135,in=0] (-0.75,-2);
            \MPU{0,-1.25}{0.75}{JungleGreen!40}{}
        \end{tikzpicture}\,.
\end{equation}

\subsection{Equivalent tensors generating the same sLPDO}
As discussed in the introduction, the tensor network representation of a state is not unique. In the setting of pure states and MPVs, we define an equivalence relation 
\begin{equation}
    A\sim_{\ket{\psi}} B
\end{equation}
between two tensors generating MPVs if the MPVs satisfy $\ket{\Psi^{(N)}(A)}=\ket{\Psi^{(N)}(B)}$ for all system sizes $N$. There is an obvious generalization to inhomogeneous tensors. An important mathematical problem is to identify the freedom among equivalent tensors. Characterizing this gauge freedom results in the fundamental theorem, which has been completed for MPSs \cite{MPS_rep,MPS_irreducible_form,MPDO_RFP,MPS_boundary}, MPUs \cite{MPU}, and for important PEPS classes \cite{PEPS_semi_injective_fundamental_theorem,PEPS_normal_fundamental_theorem}. 

In this work, we study the fundamental theorem for LPDOs. Here the relevant equivalence is a little more subtle: not equality as MPVs, but equality after tracing out the ancilla, so we write 
\begin{equation}
    A\sim_{\rho}B
\end{equation}
if purification tensors $A$ and $B$ generate the same density operator $\rho^{(N)}(A)=\rho^{(N)}(B)$ for all system sizes $N$. A slight and obvious modification is required for homogeneous sLPDOs, as the state is generated by a pair $(A,\omega)$. Requiring only $(A,\sqrt{\omega})\sim_{\ket{\psi}}(B,\sqrt{\omega'})$ identifies tensors generating the same purification, missing the additional isometric freedom in the purification degrees of freedom. From Lemma~\ref{lem:purification_partial_isometry}, given two purification tensors $A$ and $B$ and respective initial mixed states $\omega$ and $\omega'$ that generate the same sLPDO, the purifications $\sigma^{(N)}(A;\sqrt{\omega})$ and $\sigma^{(N)}(B;\sqrt{\omega'})$ must differ by a global partial isometry. The central question is: under what conditions does this global partial isometry guarantee an efficient tensor network decomposition? Such a tensor network decomposition is important since it allows tensors $A$ and $B$ to be connected locally, which is not expected to be the case in general.

Still, a positive answer to the above question under sufficient conditions constitutes only a partial answer to the full fundamental theorem for general sLPDO or LPDO purification tensors. By analogy with the fundamental theorem for MPS, one might expect that tensors failing to satisfy these sufficient conditions could first be transformed, within their equivalence class, into a suitable canonical form for which such a theorem can be established. 
Whether a canonical form for LPDO tensors even exists, is however elusive. 
Given the non-one-dimensional contraction structure of LPDOs, \ie contraction along both horizontal and vertical directions, we anticipate obstructions to a fundamental theorem analogous to those encountered for PEPS.

\section{Fundamental theorem for sequentially generated LPDOs}
\label{sec:Fundamental_theorem}
In this section, we present two independent sets of sufficient conditions for the existence of a fundamental theorem that relates tensors generating the same state. We formulate this in the class of homogeneous sLPDOs. Interestingly, even though both conditions are generic, one implies a much stronger fundamental theorem than the other.

We first introduce the canonical form (CF) for MPVs \cite{MPS_rep} under the setting of sLPDO, which will be useful for the statement of the fundamental theorem. Let $A^{ia}_{[n]}\in \mathcal{M}_{D_{n+1},D_{n}}$ be a purification tensor. A purification of an sLPDO generated by $(\{A_{[n]}\},\sqrt{\omega})$ is in CF if the following conditions are satisfied:
\begin{enumerate}
    \item $\sum_{i,a} (A^{ia}_{[n]})^{\dagger}A^{ia}_{[n]}=\id_{D_{[n]}}$ for $1\leq n\leq N$.
    \item $\sum_{i,a}A^{ia}_{[n]}\Lambda_{[n]}(A^{ia}_{[n]})^{\dagger}=\Lambda_{[n+1]}$ for $1\leq n\leq N-1$. The condition on the $N^{\mathrm{th}}$ tensor $A_{[N]}$ is that the overall state is normalized, $\tr\sum_{i,a}A^{ia}_{[N]}\Lambda_{[N]}(A^{ia}_{[N]})^\dagger=1$.
    \item $\Lambda_{[n]}>0$ is diagonal and $\tr\Lambda_{[n]}=1$ for $1\leq n\leq N$. In particular, $\Lambda_{[1]}=\omega$ is diagonal.
\end{enumerate}
A limitation of the homogeneous sLPDO ansatz is that it is not closed under a CF transformation -- that is, the CF of $\sigma^{(N)}(A;\sqrt{\omega})$ is generally a non-homogeneous $\sigma^{(N)}(\{\tilde{A}_{[n]}\};\sqrt{\tilde\omega})$. Because the final memory leg is retained as a physical degree of freedom, putting an sLPDO into CF generally requires site-dependent gauge transformations.

\subsection{Step-injective tensors}
The first condition is a local property of the purification tensor called step-injectivity. In the following, we formulate definitions and theorems for homogeneous sLPDOs, though it can be easily generalized to inhomogeneous sLPDOs by requiring that each tensor $A_{[n]}$ is step-injective.

\begin{definition}
\label{def:injectivity_sLPDO}
    A purification tensor $A$ is said to be step-injective if there exists another rank-4 tensor $A^{-1}$ such that $\sum_{i,\alpha}(A^{-1})^{ia}_{\alpha\beta}A^{ib}_{\alpha\gamma}=\delta_{ab}\delta_{\beta\gamma}$. Equivalently, when viewed as a map $A:\mathbb{C}^p\otimes \mathbb{C}^D\to\mathbb{C}^d\otimes \mathbb{C}^D$, $A$ is injective.
\end{definition}
\noindent Graphically, the step-injective condition says that
\begin{equation}
\label{eq:injectivity_sLPDO}
    \begin{tikzpicture}[baseline={([yshift=-0.5ex]current bounding box.center)},x=\cell cm,y=\cell cm]
        \draw (-1,1.5) -- (-1,0) -- (1,0);
        \draw[dashed] (0,-1) -- (0,0);
        \draw (0,0) -- (0,1.5);
        \draw[dashed] (0,2.5) -- (0,1.5);
        \MPS{0,0}{1}{blue!20}{\scriptsize $A$}
        \draw (-1,1.5) -- (1,1.5);
        \MPS{0,1.5}{1}{blue!20}{\scriptsize $A^{-1}$}
    \end{tikzpicture}=
    \begin{tikzpicture}[scale=0.6, baseline={([yshift=-0.5ex]current bounding box.center)}]
        \draw[dashed] (0,2.5) -- (0,-1);
        \draw (1,1.5) -- (0.5,1.5) -- (0.5,0) -- (1,0);
    \end{tikzpicture}\,.
\end{equation}

Injectivity is a common assumption in both the MPS \cite{MPS_rep,MPDO_RFP} and PEPS \cite{PEPS_injective} literature, because it is a useful property that establishes a sort of minimality between the virtual and physical space; it grants access to possibly contracted degrees of freedom by `opening up bonds' with the inverse tensor. In fact, a generic purification tensor $A$ is step-injective because a random matrix is expected to be full-rank; though step-injectivity bounds the Kraus rank of the quantum channels by the physical dimension $d\geq p$. Non-injective tensors are special cases and may signal the presence of underlying symmetries. In the case of sLPDOs, a step-injective tensor can be interpreted as an information-preserving interaction at each step of the sequential generation. We note that when the ancilla dimension is trivial, the definition of step-injectivity here \textit{does not} reduce to the usual injectivity of the MPS tensor -- in fact, Eq.~\eqref{eq:injectivity_sLPDO} reduces to a canonical form condition, which can always be achieved.

An immediate implication of step-injectivity of a tensor $A$ is that the purification $\sigma^{(N)}(A;\sqrt{\omega})$ is full-rank, i.e.\@ minimal, in the sense that the ancillary dimension cannot be reduced. Furthermore, the left inverse, which we shall denote by $\tau^{(N)}(A; \sqrt{\omega})$, is itself a tensor network. This is a non-trivial statement, because for generic MPOs, the inverse operator is not an MPO for any system size.

\begin{proposition}
\label{prop:inverse_MPO}
    The left inverse $\tau^{(N)}(A;\sqrt{\omega})$ of a purification $\sigma^{(N)}(A;\sqrt{\omega})$ generated by a pair $(A,\sqrt{\omega})$, with $A$ step-injective, allows an MPO representation with the same bond dimension as $A$.
    \begin{proof}
        We prove this by constructing the MPO representation of the left inverse $\tau^{(N)}(A;\sqrt{\omega})$:
        \begin{equation}
            \tau^{(N)}(A;\sqrt{\omega})=\begin{tikzpicture}[baseline={([yshift=-0.5ex]current bounding box.center)},x=\cell cm,y=\cell cm]
            \draw (-1,0.5) -- (-1,1.5) -- (1,1.5);
            \node at (1.5,1.5) {$\cdots$};
            \draw (2,1.5) -- (4.25,1.5) -- (4.25,2.5);
            \draw (0,1.5) -- (0,0.5);
            \draw (3,1.5) -- (3,0.5);
            \draw[dashed] (0,2.5) -- (0,1.5);
            \draw[dashed] (3,2.5) -- (3,1.5);
            \MPS{0,1.5}{1}{blue!20}{\scriptsize $A^{-1}$}
            \MPS{3,1.5}{1}{blue!20}{\scriptsize $A^{-1}$}
            \MPU{4.25,1.5}{0.75}{green!20}{\scriptsize $K$}
        \end{tikzpicture}
        \end{equation}
        where $K$ is a left inverse $K\sqrt{\omega}=\id_{\mathbb{C}^r}$.  Such a left inverse exists because $\sqrt{\omega}$ has full column rank. By construction and the step-injectivity condition, it satisfies $\tau^{(N)}(A;\sqrt{\omega})\sigma^{(N)}(A;\sqrt{\omega})=\id$.
    \end{proof}
\end{proposition}
We now present the fundamental theorem for sLPDOs generated by step-injective tensors. By Corollary~\ref{cor:purification_freedom}, any two purifications can be connected by an isometry; the step-injectivity condition guarantees that this isometry admits a matrix product form.

\begin{theorem}
\label{thm:sLPDO-FT-1}
    Let $(A,\omega)$ be a pair generating sLPDO with $A$ step-injective. Let $(B,\omega')$ be another pair generating the same sLPDO. Then the unique isometry $U_N$ connecting purifications of Lemma~\ref{lem:purification_partial_isometry} admits an MPI representation and
    \begin{equation}
        \left(\left\{\begin{tikzpicture}[baseline={([yshift=-0.5ex]current bounding box.center)},x=\cell cm,y=\cell cm]
            \draw (1,0) -- (-1,0);
            \draw (-1,-1.25) -- (1,-1.25);
            \draw (0,0) -- (0,1);
            \draw[dashed] (0,-2) -- (0,0);
            \MPS{0,0}{1}{blue!20}{\scriptsize $A$}
            \MPU{0,-1.25}{1}{purple!20}{\scriptsize $U_{[n]}$}
        \end{tikzpicture}\right\},\begin{tikzpicture}[baseline={([yshift=-0.5ex]current bounding box.center)},x=\cell cm,y=\cell cm]
            \draw (-0.75,1.25) -- (0,1.25) -- (0,0) -- (-0.75,0);
            \draw (-0.75,-1.25) -- (0,-1.25) -- (0,-2);
            \MPU{0,1.25}{0.75}{green!20}{\scriptsize $\sqrt{\omega}$}
            \MPU{0,0}{0.75}{green!20}{\scriptsize $K$}
            \MPU{0,-1.25}{0.75}{green!20}{\scriptsize $\sqrt{\omega'}$}
        \end{tikzpicture}\right) \sim_{\ket{\psi}} \left(\begin{tikzpicture}[baseline={([yshift=-0.5ex]current bounding box.center)},x=\cell cm,y=\cell cm]
            \draw (1,0) -- (-1,0);
            \draw (0,0) -- (0,1);
            \draw[dashed] (0,-1) -- (0,0);
            \MPS{0,0}{1}{blue!20}{\scriptsize $B$}
        \end{tikzpicture},\begin{tikzpicture}[baseline={([yshift=-0.5ex]current bounding box.center)},x=\cell cm,y=\cell cm]
            \draw (-0.75,-1.25) -- (0,-1.25) -- (0,-2);
            \MPU{0,-1.25}{0.75}{green!20}{\scriptsize $\sqrt{\omega'}$}
        \end{tikzpicture}\right)
    \end{equation}
    where $U$ is the tensor generating the MPI $U_N$. Concretely, 
    \begin{align}
    \begin{split}
        \begin{tikzpicture}[baseline={([yshift=-0.5ex]current bounding box.center)},x=\cell cm,y=\cell cm]
            \draw (-1,0) -- (1,0);
            \draw[dashed] (0,-1) -- (0,1);
            \MPU{0,0}{1}{purple!20}{\scriptsize $U_{[n]}$}
        \end{tikzpicture}&=
        \begin{tikzpicture}[baseline={([yshift=-0.5ex]current bounding box.center)},x=\cell cm,y=\cell cm]
            \draw (0,0) -- (0,1);
            \draw[dashed] (0,-1) -- (0,0);
            \draw[dashed] (0,2.5) -- (0,1.5);
            \draw (-1,0) -- (1,0);
            \draw (-1,1.5) -- (1,1.5);
            \MPS{0,0}{1}{blue!20}{\scriptsize $B$}
            \MPS{0,1.5}{1}{blue!20}{\scriptsize $A^{-1}$}
        \end{tikzpicture}\,, \qquad \text{for } 1\leq n\leq N-1\,, \\
        \\
        \begin{tikzpicture}[baseline={([yshift=-0.5ex]current bounding box.center)},x=\cell cm,y=\cell cm]
            \draw (-1,0) -- (1,0);
            \draw[dashed] (0,-1) -- (0,1);
            \MPU{0,0}{1}{purple!20}{\scriptsize $U_{[N]}$}
        \end{tikzpicture}&=
        \begin{tikzpicture}[baseline={([yshift=-0.5ex]current bounding box.center)},x=\cell cm,y=\cell cm]
            \draw (0,0) -- (0,1);
            \draw[dashed] (0,-1) -- (0,0);
            \draw[dashed] (0,2.5) -- (0,1.5);
            \draw (-1,1.5) -- (-1,0) -- (1,0);
            \draw (-1,1.5) -- (1,1.5);
            \MPS{0,0}{1}{blue!20}{\scriptsize $B$}
            \MPS{0,1.5}{1}{blue!20}{\scriptsize $A^{-1}$}
        \end{tikzpicture}\,,
    \end{split}
    \end{align}
    which slightly breaks homogeneity.
    
    If additionally $(\{B_{[n]}\},\sqrt{\omega'})$ is in CF, then there exist rectangular matrices $\{Y_{[n]}\}_{n=0}^{N-1}$ and $\{Z_{[n]}\}_{n=1}^N$ with $Y_{[n]}Z_{[n+1]} =\id$ such that 
    \begin{align}
    \begin{split}
            \begin{tikzpicture}[baseline={([yshift=-0.5ex]current bounding box.center)},x=\cell cm,y=\cell cm]
            \draw (-0.75,1.25) -- (0,1.25) -- (0,0) -- (-0.75,0);
            \draw (-0.75,-1.25) -- (0,-1.25) -- (0,-2);
            \draw (-1,0) -- (-2,0);
            \MPU{0,1.25}{0.75}{green!20}{\scriptsize $\sqrt{\omega}$}
            \MPU{0,0}{0.75}{green!20}{\scriptsize $K$}
            \MPU{0,-1.25}{0.75}{green!20}{\scriptsize $\sqrt{\omega'}$}
            \transfermat{-1,0}{3}{0.75}{orange!20}{\scriptsize $Y_{[0]}$}
        \end{tikzpicture}=\begin{tikzpicture}[baseline={([yshift=-0.5ex]current bounding box.center)},x=\cell cm,y=\cell cm]
            \draw (-0.75,-1.25) -- (0,-1.25) -- (0,-2);
            \MPU{0,-1.25}{0.75}{green!20}{\scriptsize $\sqrt{\omega'}$}
        \end{tikzpicture}\,, \qquad \begin{tikzpicture}[baseline={([yshift=-0.5ex]current bounding box.center)},x=\cell cm,y=\cell cm]
            \draw (1,0) -- (-1,0) -- (-1,1);
            \draw (0,-1.25) -- (1,-1.25);
            \draw (0,0) -- (0,1);
            \draw[dashed] (0,-2) -- (0,0);
            \draw (1,-0.625) -- (2,-0.625);
            \MPS{0,0}{1}{blue!20}{\scriptsize $A$}
            \MPU{0,-1.25}{1}{purple!20}{\scriptsize $U_{[N]}$}
            \transfermat{1.25,-0.625}{2}{0.75}{Cyan!20}{\scriptsize $Z_{[N]}$}
        \end{tikzpicture}=\begin{tikzpicture}[baseline={([yshift=-0.5ex]current bounding box.center)},x=\cell cm,y=\cell cm]
            \draw (1,0) -- (-1,0);
            \draw (0,0) -- (0,1);
            \draw[dashed] (0,-1) -- (0,0);
            \MPS{0,0}{1}{blue!20}{\scriptsize $B_{[N]}$}
        \end{tikzpicture}\,, \\
        \begin{tikzpicture}[baseline={([yshift=-0.5ex]current bounding box.center)},x=\cell cm,y=\cell cm]
            \draw (1,0) -- (-1,0);
            \draw (-1,-1.25) -- (1,-1.25);
            \draw (0,0) -- (0,1);
            \draw[dashed] (0,-2) -- (0,0);
            \draw (-1,-0.625) -- (-2,-0.625);
            \draw (1,-0.625) -- (2,-0.625);
            \MPS{0,0}{1}{blue!20}{\scriptsize $A$}
            \MPU{0,-1.25}{1}{purple!20}{\scriptsize $U_{[n]}$}
            \transfermat{-1.25,-0.625}{2}{0.75}{orange!20}{\scriptsize $Y_{[n]}$}
            \transfermat{1.25,-0.625}{2}{0.75}{Cyan!20}{\scriptsize $Z_{[n]}$}
        \end{tikzpicture}=\begin{tikzpicture}[baseline={([yshift=-0.5ex]current bounding box.center)},x=\cell cm,y=\cell cm]
            \draw (1,0) -- (-1,0);
            \draw (0,0) -- (0,1);
            \draw[dashed] (0,-1) -- (0,0);
            \MPS{0,0}{1}{blue!20}{\scriptsize $B_{[n]}$}
        \end{tikzpicture}\,, \qquad \text{for } 1\leq n\leq N-1.
    \end{split}
    \end{align}
    \begin{proof}
        The proof is a straightforward application of corollary~\ref{cor:purification_freedom} and proposition~\ref{prop:inverse_MPO}. Since $A$ is step-injective, $\sigma^{(N)}(A;\sqrt{\omega})$ is minimal. For each system size $N$, there exists a unique isometry $U_N:(\mathbb{C}^{p_B})^{\otimes N}\otimes \mathbb{C}^{r_B}\to (\mathbb{C}^{p_A})^{\otimes N}\otimes \mathbb{C}^{r_A}$ with $U_N U^\dagger_N =\id$ and $ U^\dagger_N U_N=\mathrm{Proj}_{\ker(\sigma^{(N)}(B))^\perp}$ such that
        \begin{equation}
            \begin{tikzpicture}[baseline={([yshift=-0.5ex]current bounding box.center)},x=\cell cm,y=\cell cm]
            \draw (1,0) -- (-1,0) -- (-1,1);
            \node at (1.5,0) {$\cdots$};
            \draw (2,0) -- (4.25,0) -- (4.25,-2.5);
            \draw (0,0) -- (0,1);
            \draw[dashed] (0,-2.5) -- (0,0);
            \draw (3,0) -- (3,1);
            \draw[dashed] (3,-2.5) -- (3,0);
            \MPS{0,0}{1}{blue!20}{\scriptsize $A$}
            \MPS{3,0}{1}{blue!20}{\scriptsize $A$}
            \draw[fill=purple!20] (-0.5,-2) -- (-0.5,-1) -- (4.75,-1) -- (4.75,-2) -- cycle;
            \node at (2.125,-1.5) {\footnotesize $U_N$};
            \MPU{4.25,0}{0.75}{green!20}{\scriptsize $\sqrt{\omega}$}
        \end{tikzpicture}=\begin{tikzpicture}[baseline={([yshift=-0.5ex]current bounding box.center)},x=\cell cm,y=\cell cm]
            \draw (1,0) -- (-1,0) -- (-1,1);
            \node at (1.5,0) {$\cdots$};
            \draw (2,0) -- (4.25,0) -- (4.25,-1);
            \draw (0,0) -- (0,1);
            \draw[dashed] (0,-1) -- (0,0);
            \draw (3,0) -- (3,1);
            \draw[dashed] (3,-1) -- (3,0);
            \MPS{0,0}{1}{blue!20}{\scriptsize $B$}
            \MPS{3,0}{1}{blue!20}{\scriptsize $B$}
            \MPU{4.25,0}{0.75}{green!20}{\scriptsize $\sqrt{\omega'}$}
        \end{tikzpicture}\,.
        \end{equation}
        We can then invert $\sigma^{(N)}(A;\sqrt{\omega})$ using $\tau^{(N)}(A;\sqrt{\omega})$, resulting in an explicit form for the isometry $U_N$:
        \begin{equation}
            \begin{tikzpicture}[baseline={([yshift=-0.5ex]current bounding box.center)},x=\cell cm,y=\cell cm]
            \draw (4,-0.5) -- (4,-2.5);
            \draw[dashed] (0,-0.5) -- (0,-2.5);
            \draw[dashed] (3,-0.5) -- (3,-2.5);
            \draw[fill=purple!20] (-0.5,-2) -- (-0.5,-1) -- (4.5,-1) -- (4.5,-2) -- cycle;
            \node at (2,-1.5) {\footnotesize $U_N$};
        \end{tikzpicture}=\begin{tikzpicture}[baseline={([yshift=-0.5ex]current bounding box.center)},x=\cell cm,y=\cell cm]
            \draw (1,0) -- (-1,0) -- (-1,1.5) -- (1,1.5);
            \node at (1.5,1.5) {$\cdots$};
            \draw (2,1.5) -- (4.25,1.5) -- (4.25,2.5);
            \node at (1.5,0) {$\cdots$};
            \draw (2,0) -- (4.25,0) -- (4.25,-1);
            \draw (0,0) -- (0,1);
            \draw[dashed] (0,-1) -- (0,0);
            \draw (3,0) -- (3,1);
            \draw[dashed] (3,-1) -- (3,0);
            \draw[dashed] (0,2.5) -- (0,1.5);
            \draw[dashed] (3,2.5) -- (3,1.5);
            \MPU{4.25,1.5}{0.75}{green!20}{\scriptsize $K$}
            \MPS{0,0}{1}{blue!20}{\scriptsize $B$}
            \MPS{3,0}{1}{blue!20}{\scriptsize $B$}
            \MPS{0,1.5}{1}{blue!20}{\scriptsize $A^{-1}$}
            \MPS{3,1.5}{1}{blue!20}{\scriptsize $A^{-1}$}
            \MPU{4.25,0}{0.75}{green!20}{\scriptsize $\sqrt{\omega'}$}
        \end{tikzpicture}\,.
        \end{equation}
        This tensor has the form of a (non-homogeneous) MPI with bond dimension $D_U=D^2$, \ie an MPO that generates an isometry for all system sizes. 

        The second statement of the theorem follows from the OBC fundamental theorem of Ref.~\cite[Theorem 2]{MPS_rep}, treating the sLPDO ansatz as an OBC MPS with $N+1$ sites (viewing the initial state $\sqrt{\omega}$ as the extra site). 
    \end{proof}
\end{theorem}

To summarize, with the step-injectivity condition, one arrives at a fundamental theorem for sLPDOs where the two purifications of the same sLPDO are guaranteed to be connected by a matrix product isometry with a finite bond dimension. 

\subsection{Cyclic tensors}
We introduce a second, global condition on the purification tensor. Since it is most naturally expressed in terms of the channel induced by $A$, we introduce some notation. For each pair of physical indices $i,j$, write $\mathcal{E}^{ij}_A(\rho)=\sum_{a}A^{ia}\rho (A^{ja})^\dagger$. Thus $\mathcal{E}^{ij}_A$ is the $(i,j)$-component of the output of $\mathcal{E}_A$.

For each pair $(A,\omega)$, we associate an algebra of maps and a vector space. For a word $w=((i_1,j_1),\dots,(i_\ell,j_\ell))$ of length $\ell$, define $\mathcal{E}^w_A=\mathcal{E}^{i_\ell j_\ell}_A\circ\cdots\circ \mathcal{E}^{i_1 j_1}_A$. We denote $\abs{w}$ as the length of the word, and define 
\begin{equation}
\label{eq:algebra_of_maps}
    \mathcal{A}_A=\mathrm{span}\{\mathcal{E}^w_A\mid \abs{w}\geq 1\}\subseteq \mathrm{End}(\mathcal{M}_D)\,,
\end{equation}
the algebra of all linear combinations of compositions of $\mathcal{E}^{ij}_A$ of arbitrary length. We also define the reachable space
\begin{equation}
    \mathcal{R}_{(A,\omega)}=\mathcal{A}_A\cdot \omega=\mathrm{span}\{\mathcal{E}^w_A(\omega)\mid \abs{w}\geq1\}\subseteq \mathcal{M}_D\,.
\end{equation}

\begin{definition}
    Let $(A,\omega)$ be a pair generating sLPDO. We say that $(A,\omega)$ is cyclic, or $\omega$ is cyclic for $A$, if $\mathcal{R}_{(A,\omega)}=\mathcal{M}_D$.
\end{definition}
This definition of cyclicity is equivalent to one of the two conditions of the previous notion of minimality appearing in the original finitely-correlated state construction~\cite[Proposition~2.1]{FNW}. In that setting, the state has both a left and right boundary, and the minimality condition corresponds to satisfying cyclicity with respect to each boundary. By constrast, the memory is not traced out in the sLPDO construction, so only cyclicity from the right boundary is required. Eq.~\eqref{eq:algebra_of_maps} is the same as the algebra of matrices studied recently in~\cite{MPS_boundary}. In any case, cyclicity of $(A,\omega)$ is weaker than the statement $\mathcal{A}_A=\mathrm{End}(\mathcal{M}_D)$.

Given a pair $(A,\omega)$, it is efficient to check whether it is cyclic. Define
\begin{equation}
    \mathcal{R}^{(\leq \ell)}_{(A,\omega)}=\mathrm{span}\{\mathcal{E}^w_A(\omega)\mid 1\leq \abs{w}\leq \ell\}\subseteq \mathcal{M}_D\,.
\end{equation}
We note two properties of these spaces. First, that $\mathcal{R}^{(\leq \ell)}_{(A,\omega)}$ form a filtration of $\mathcal{R}_{(A,\omega)}$, that is, $\mathcal{R}_{(A,\omega)}=\bigcup_{\ell\geq 1}\mathcal{R}^{(\leq \ell)}_{(A,\omega)}$ and $\mathcal{R}^{(\leq \ell)}_{(A,\omega)}\subseteq \mathcal{R}^{(\leq \ell')}_{(A,\omega)}$ for $\ell\leq \ell'$. Furthermore, if $\mathcal{R}^{(\leq \ell)}_{(A,\omega)}=\mathcal{R}^{(\leq \ell+1)}_{(A,\omega)}$ for some $\ell$, then $\mathcal{R}^{(\leq \ell+k)}_{(A,\omega)}=\mathcal{R}^{(\leq\ell)}_{(A,\omega)}$ for any $k\geq 0$. These subspaces must stabilize after at most $D^2$ steps. Thus, it is sufficient to check whether $\mathcal{R}^{\leq D^2}_{(A,\omega)}=\mathcal{M}_D$.

Now, we show that the cyclic condition has a strong consequence on the fundamental theorem of sLPDOs.
\begin{lemma}
\label{lem:equal_on_words}
    Let $(A,\omega)$ and $(B,\omega')$ be pairs generating the same sLPDO. Then $\mathcal{E}^w_A(\omega)=\mathcal{E}^w_B(\omega')$ for all words $w$. In particular, $\mathcal{R}_{(A,\omega)}=\mathcal{R}_{(B,\omega')}$.
    \begin{proof}
        We rewrite the sLPDO generated by $(A,\omega)$ as
        \begin{equation}
            \rho^{(N)}(A;\omega)=\sum_{\substack{i_1,\dots,i_N \\ j_1,\dots,j_N}}\mathcal{E}^w_{A}(\omega)\otimes\ketbra{i_N,\dots,i_1}{j_N,\dots,j_1}\,,
        \end{equation}
        where $w=((i_1,j_1),\dots,(i_N,j_N))$, and similarly for $B$. The equality of states for all $N$ implies $\mathcal{E}^w_A(\omega)=\mathcal{E}^w_B(\omega')$ for all words $w$, from which the claim follows.
    \end{proof}
\end{lemma}
\noindent In this case, we will denote both reachable spaces by $\mathcal{R}$. What this lemma implies is actually that cyclicity is not a property of the sLPDO representation, but rather a property of the state itself, because all representations have the same reachable space.

\begin{theorem}
    Let $(A,\omega)$ and $(B,\omega')$ be cyclic pairs generating the same sLPDO. Then $B^{ia}=\sum_b U_{ab}A^{ib}$ for some isometry $U$, meaning
    \begin{equation}
        \begin{tikzpicture}[baseline={([yshift=-0.5ex]current bounding box.center)},x=\cell cm,y=\cell cm]
            \draw (1,0) -- (-1,0);
            \draw (0,0) -- (0,1);
            \draw[dashed] (0,-1.75) -- (0,0);
            \MPS{0,0}{1}{blue!20}{\scriptsize $A$}
            \MPU{0,-1.125}{0.75}{purple!20}{\scriptsize $U$}
        \end{tikzpicture}=\begin{tikzpicture}[baseline={([yshift=-0.5ex]current bounding box.center)},x=\cell cm,y=\cell cm]
            \draw (1,0) -- (-1,0);
            \draw (0,0) -- (0,1);
            \draw[dashed] (0,-1) -- (0,0);
            \MPS{0,0}{1}{blue!20}{\scriptsize $B$}
        \end{tikzpicture}\,.
    \end{equation}
    \begin{proof}
        First, we observe that $\mathcal{E}^{ij}_A$ and $\mathcal{E}^{ij}_B$ have equal action on $\mathcal{R}$. Indeed, let $X\in\mathcal{R}$ and write $X=\sum_{w} c_w \mathcal{E}^w_A(\omega)$. Then by lemma \ref{lem:equal_on_words},
        \begin{equation}
            \mathcal{E}^{ij}_A(X)=\sum_w c_w \mathcal{E}^{w(ij)}_A(\omega)=\sum_wc_w\mathcal{E}^{w(ij)}_B(\omega')=\mathcal{E}^{ij}_B(X)\,.
        \end{equation}
        By cyclicity, $\mathcal{R}=\mathcal{M}_D$, so $\mathcal{E}^{ij}_A=\mathcal{E}^{ij}_B$ as linear maps. Thus $\mathcal{E}_A=\mathcal{E}_B$ as CP maps. The freedom in the Kraus representation gives an on-site isometry relating purification tensors in the ancilla index.
    \end{proof}
\end{theorem}
To summarize, with the cyclic condition, the two purifications of the same sLPDO are guaranteed to be connected by an on-site isometry. 

It is interesting that the cyclic condition implies a much stronger fundamental theorem than the one for step-injective tensors, that the tensors are related by an on-site isometry in the ancillary space. These two conditions are independent and are neither necessary nor sufficient for each other -- we present some examples shortly in section~\ref{sec:examples}. Both conditions are expected to be generic. 

\section{Examples of sLPDOs}
\label{sec:examples}
Having presented the two fundamental theorems of sLPDO under two different conditions, we now
provide classes of interesting and relevant quantum states that fall under the sLPDO ansatz. We also comment on their step-injective and cyclic properties in light of the fundamental theorems presented in the previous section.

\subsection{Boundaries of topological theories}
A useful tool to study topologically ordered phases~\cite{levin2005string-net} is the holographic encoding of a bulk theory into a boundary theory of one dimension lower. In the case of two-dimensional PEPS, the boundary theory arises naturally on the virtual space exposed by a bipartite cut \cite{PEPS_entanglement_Hamiltonian,PEPS_TO_transfer_matrix}. This gives rise to an important class of one-dimensional mixed states consisting of boundaries of two-dimensional topologically-ordered systems. We show that the sLPDO ansatz is sufficiently general to capture the boundary of $D(G)$ topological order, where $G$ is a finite group, offering a new perspective in which such boundary theories are understood as those which are sequentially preparable mixed states without postselection.

The construction of a PEPS realizing the Kitaev quantum double $D(G)$ is detailed in \cite{PEPS_Ginjective}. The PEPS tensor itself obeys a local symmetry, and it is this symmetry that implies that the boundary on the virtual space of such a state is simply a projector onto the $G$-invariant subspace, equivalently expressed as
\begin{equation}
    \rho^{(N)}_G=\frac{1}{\abs{G}^N}\sum_{h\in G}\sum_{g_1,\dots,g_N\in G}\ketbra{hg_N,\dots,hg_1}{g_N,\dots,g_1}.
\end{equation}
Note that the summands are precisely the matrix representation of the site-wise left multiplication operator $L_h^{\otimes N}$ defined by $L_h\ket{g}=\ket{hg}$, so the state can also be written as 
\begin{equation}
\label{eq:rho_G_multiplication}
    \rho^{(N)}_G=\frac{1}{\abs{G}^N}\sum_{h\in G}L_h^{\otimes N}\,.
\end{equation} For details of this construction, we refer to appendix~\ref{app:boundary}. 

The state $\rho^{(N)}_G$ admits an exact representation as a bond dimension $D=\abs{G}$ sLPDO with purification tensor $A^{ia}_{\alpha\beta}=\frac{1}{\sqrt{\abs{G}}}\delta^{i,\alpha a}\delta_{\alpha\beta}$ with $i,a,\alpha,\beta\in G$, and initial state $\omega=\ketbra{+_G}{+_G}$, where $\ket{+_G}=\frac{1}{\sqrt{\abs{G}}}\sum_{x}\ket{x}$. To see this, the sequential channel is calculated to be $\mathcal{E}_A(\ketbra{x}{y})=\frac{1}{\abs{G}}\sum_{g}\ketbra{x}{y} \otimes\ketbra{xg}{yg}$. Iterating this $N$ times yields the state
\begin{equation}
\begin{aligned}
    \rho^{(N)}(A;\omega)&=\frac{1}{\abs{G}^{N+1}}\sum_{x,y}\sum_{g_1,\dots,g_N} \ketbra{x}{y}\otimes \\
    &\phantom{={}}\ketbra{xg_N,\dots,xg_1}{yg_N,\dots,yg_1}\,.
\end{aligned}
\end{equation}
from which the result follows after relabelling $y=g_0$ and $h=xy^{-1}$. Intuitively, the memory carries a fixed group element and the channel applies it uniformly to each physical index. Thus for any system size $N$, the state generated on $N$ qu-$\abs{G}$-its and memory is the boundary of the $D(G)$ topological order on $N+1$ sites, $\rho^{(N)}(A;\omega)=\rho^{(N+1)}_G$.

We note that the tensor $A$ generating the boundary of $D(G)$ topological order is step-injective with inverse $(A^{-1})^{ia}_{\alpha\beta}=\sqrt{\abs{G}}\delta^{\alpha^{-1}i,a}\delta_{\alpha,\beta}$. However, starting from the state $\omega=\ketbra{+_G}{+_G}$, $(A,\omega)$ is not cyclic, because it only explores the space spanned by the left multiplication operator, $\mathcal{R}_{(A,\omega)}=\mathrm{span}\{L_h=\sum_{y}\ketbra{hy}{y}\mid h\in G\}$. This is a vector space of dimension $\dim \mathcal{R}_{(A,\omega)}=\abs{G}<\abs{G}^2=\dim \mathcal{M}_{\abs{G}}$.

In the case of an abelian group $G$, there is an alternative choice of basis in which $\rho^{(N)}_G$ is diagonal. We briefly outline the construction in appendix~\ref{app:abelian_boundary_TO}, relying on the properties of characters of finite abelian groups. The upshot is that after a suitable basis change, the state can be written
\begin{equation}
    \tilde{\rho}^{(N)}_G:=\frac{1}{\abs{G}^{N-1}}\sum_{\substack{i_1,\dots,i_N\in G \\ i_1\cdots i_N=e}}\ketbra{i_N,\dots, i_1}{i_N,\dots, i_1}\,,
\end{equation}
an incoherent uniform mixture over neutral-charge strings. Here $e\in G$ is the multiplicative unit.

As an explicit example, we consider the boundary of the toric code, with
\begin{equation}
    \tilde{\rho}_{\mathbb{Z}_2}^{(N)}=\frac{1}{2^N}(I^{\otimes N}+Z^{\otimes N})\,.
\end{equation}
It is generated by the tensor $A^{00}=\frac{1}{\sqrt{2}}I$, $A^{01}=0$, $A^{10}=0$, $A^{11}=\frac{1}{\sqrt{2}}X$ and initial state $\omega=\ketbra{0}{0}$. The action of the channel is given by $\mathcal{E}^{00}_A(\rho)=\frac{1}{2}\rho$, $\mathcal{E}^{11}_A(\rho)=\frac{1}{2}X\rho X$.

There exists another purification tensor $B$ generating the toric code boundary, with ancilla dimension $p_B=4$. Its non-zero components are $B^{00}=\frac{1}{\sqrt{2}}\ketbra{0}{0}$, $B^{01}=\frac{1}{\sqrt{2}}\ketbra{1}{1}$, $B^{12}=\frac{1}{\sqrt{2}}\ketbra{1}{0}$ and $B^{13}=\frac{1}{\sqrt{2}}\ketbra{0}{1}$. Before saying which initial state $\omega'$ it acts on, we first note that the action of the channel on the memory state is $\mathcal{E}^{00}_B(\rho)=\frac{1}{2}D(\rho)$ and $\mathcal{E}^{11}_B(\rho)=\frac{1}{2}XD(\rho) X$, where $D(\rho)$ is the completely dephasing channel $D(\rho)=\mel{0}{\rho}{0}\ketbra{0}{0}+\mel{1}{\rho}{1}\ketbra{1}{1}$. Therefore, it can generate the toric code boundary from any initial state with the property that $\mel{1}{\omega'}{1}=0$, which, in the 2-dimensional case, forces $\omega'=\ketbra{0}{0}$. This tensor is neither step-injective nor cyclic, but is related to the tensor $A$ above by an MPI by theorem~\ref{thm:sLPDO-FT-1}. 

\subsection{Reduced translationally-invariant MPS states}
All sLPDOs can be interpreted as a reduced state of an MPS tensor, by writing a factorization of the purification tensor $A^{ia}=B^i C^a$. Restricting to the cases where $B=C$ forms a subclass of homogeneous sLPDOs:
\begin{equation}
    \begin{tikzpicture}[baseline={([yshift=-0.5ex]current bounding box.center)},x=\cell cm,y=\cell cm]
        \draw (-1,0) -- (1,0);
        \draw[dashed] (0,-1) -- (0,0);
        \draw (0,0) -- (0,1);
        \MPS{0,0}{1}{blue!20}{\scriptsize $A$} 
    \end{tikzpicture}
    =\begin{tikzpicture}[baseline={([yshift=-0.5ex]current bounding box.center)},x=\cell cm,y=\cell cm]
        \draw (-1,0) -- (2.5,0);
        \draw (0,0) -- (0,1);
        \draw[densely dashed] (1.5,0) -- (1.5,-1);
        \MPS{0,0}{1}{blue!20}{\scriptsize $C$} 
        \MPS{1.5,0}{1}{blue!20}{\scriptsize $C$}
    \end{tikzpicture}\,.
\end{equation}

First, let $C^i\in\mathcal{M}_D$ be a tensor generating MPS. The definition $A^{ia}=C^iC^a$ implies that $\mathrm{rank}(A)\leq D$, as viewed as a map $A:\mathbb{C}^d\otimes \mathbb{C}^D\to \mathbb{C}^d\otimes \mathbb{C}^D$. In particular, because step-injectivity requires full rank, no reduced translationally-invariant MPS state can be step-injective in the sense of definition~\ref{def:injectivity_sLPDO}.

A sufficient condition for cyclicity of $A$ with any non-zero initial state $\omega$ is the injectivity of $C$ as an MPS. Recall that injectivity for an MPS means that $\mathrm{span}\{C^i\}=\mathcal{M}_D$. We claim that if $C$ is an injective tensor, then $\mathrm{span}\{\mathcal{E}_A^{ij}(\omega)\mid i,j\in\{0,\dots,d-1\}\}=\mathcal{M}_D$ for any non-zero $\omega\geq 0$. To show this, it is sufficient to show that any rank-1 element $\ketbra{u}{v}\in\mathcal{M}_D$ lies within the span. Define $\mathcal{E}_C(\rho)=\sum_i C^i\rho(C^i)^\dagger$ to be the transfer matrix of $C$ acting on $\omega$. Take any non-zero matrix element $\mel{x}{\mathcal{E}_C(\omega)}{y}$ of the transfer matrix and consider $|u\rangle\!\!\mel{x}{\mathcal{E}_C(\omega)}{y}\!\!\langle v|$.
Note that $\ketbra{u}{x},\ketbra{v}{y}\in\mathrm{span}\{C^i\}$. Hence the span of $\mathcal{E}^{ij}_A(\omega)=C^i\mathcal{E}_C(\omega)(C^j)^\dagger$ contains $\ketbra{u}{v}$, showing that $\mathcal{R}^{(\leq 1)}_{(A,\omega)}=\mathcal{M}_D$ and that the pair $(A,\omega)$ is cyclic. 

\section{Beyond sLPDOs}
\label{sec:counterexample}
Having discussed a class of tensors where a fundamental theorem was successful, we explore examples of LPDOs which go beyond the sLPDO ansatz. We present an important example which illustrates that two valid purifications cannot always be connected by MPIs. Both purifications are minimal, suggesting that any hope for a fundamental theorem relating the tensors locally requires a further reduction in terms of a canonical form, or is outright impossible. In any case, these instructive examples help us to understand the obstructions to a full fundamental theorem for PBC LPDOs.

In appendix~\ref{app:CCDO}, we additionally investigate a class of uniform LPDOs with diagonal boundary, those forming convex combinations of MPS pure states. We identify a subclass of these tensors for which a full fundamental theorem can be proved.

\subsection{Counterexamples to MPI connecting purifications}
We present a class of uniform LPDOs (\ie of type Eq.\@~\eqref{eq:LPDO} with PBC) generating the same density matrix, $\rho^{(N)}(A)=\rho^{(N)}(B)$, yet the unitary connecting the purifications can be shown to have Schmidt rank that is exponential in system size across some bipartition. This is sufficient to show that such a unitary cannot be an MPU. The strategy will be to pick a purification that is normal, \ie $[\sigma^{(N)}(A),\sigma^{(N)}(A)^\dagger]=0$ and $\sigma^{(N)}(B)=\sigma^{(N)}(A)^\dagger$. We achieve this by choosing a block-diagonal purification tensor 
\begin{equation}
    A^{ia}=\begin{pmatrix}
    \delta^{ia} & 0 \\
    0 & \epsilon W^{ia}
\end{pmatrix}\,,
\end{equation} 
where $W^{ia}_{\alpha\beta}=\delta^i_\alpha \delta^a_\beta$ is the tensor generating the right shift operator $S$. Here tensor $A$ has dimension $d=p$ and bond dimension $D=d+1$. The purification generated by $A$ is $\sigma^{(N)}(A)=\id+\epsilon^N\, S$, and taking $B^{ia}=\overline{A}{}^{ai}$ the one for $B$ is $\sigma^{(N)}(B)=\id+(\epsilon^*)^N\, S^\dagger$. The density operator generated is $\rho^{(N)}(A)=(1+\abs{\epsilon}^{2N})\id+\epsilon^NS+(\epsilon^*)^NS^\dagger$.

When $\abs{\epsilon}\neq1$, the purifications are invertible and the inverse is given by $\sigma^{(N)}(A)^{-1}=\frac{1}{1-(-\epsilon^N)^N}\sum_{k=0}^{N-1}(-\epsilon^N)^kS^k$. The unitary relating the purifications is unique and given by
\begin{align}
\label{eq:counterexample_coefficients}
\begin{split}
    U=\sigma^{(N)}(A)^{-1}\sigma^{(N)}(B)=\sum_{k=0}^{N-1}u_kS^k\,, \\
    u_k=\begin{cases}
        \frac{(-\epsilon^N)^k(1-\abs{\epsilon}^{2N})}{1-(-\epsilon^N)^N}\,, & k=0,\dots,N-2 \\
        \frac{(-\epsilon^N)^{N-1}+(\epsilon^*)^N}{1-(-\epsilon^N)^N}\,, & k=N-1
    \end{cases}\,.
\end{split}
\end{align}
The rest of the section is devoted to showing that such a unitary has Schmidt rank growing exponentially in $N$ at half-bipartition. After, we analyze the $\abs{\epsilon}=1$ case.

Suppose for simplicity that $N$ is even and $N\geq 4$. From~\eqref{eq:counterexample_coefficients}, the $u_{N/2}$ term is non-zero. We will use the exponential Schmidt rank of $S^{N/2}$ to lower bound the Schmidt rank of $U$. We judiciously define subspaces $V_{\text{in}},V_{\text{out}}\subseteq(\mathbb{C}^{d})^{\otimes N}$ by
\begin{align}
\begin{split}
    V_{\text{in}}=\mathrm{span}\{|\underbrace{0\cdots 0}_{N/2}1\vb{a}1\rangle\mid \vb{a}\in\{0,1,\dots d-1\}^{N/2-2}\}\,, \\
    V_{\text{out}}=\mathrm{span}\{|1\vb{a}1\underbrace{0\cdots 0}_{N/2}\rangle\mid \vb{a}\in\{0,1,\dots d-1\}^{N/2-2}\}\,,
\end{split}
\end{align}
which have the property that $P_{\text{out}}S^kP_{\text{in}}=0$ when $k\neq N/2$, where $P_{\text{in}/\text{out}}$ is the projection onto $V_{\text{in}/\text{out}}$. In particular, $P_{\text{out}}UP_{\text{in}}=u_{N/2}P_{\text{out}}S^{N/2}P_{\text{in}}$. We compute
\begin{equation}
    P_{\text{out}}S^{N/2}P_{\text{in}}=\hspace{-5mm}\sum_{\vb{a}\in\{0,1,\dots d-1\}^{N/2-2}}\hspace{-5mm}\ketbra{1\vb{a}1}{0}^{\otimes N/2}\otimes |0\rangle^{\otimes N/2}\langle 1\vb{a}1|\,,
\end{equation}
which is already a Schmidt decomposition with $d^{N/2-2}$ terms. Thus $\mathrm{SR}(P_{\text{out}}UP_{\text{in}})=d^{N/2-2}$. 

The crucial observation now is that $P_{\text{in}/\text{out}}$ factorizes across the half-chain bipartition; let us denote it by $P_{\text{in}}=P_0\otimes P_{1\vb{a}1}$ and $P_{\text{out}}=P_{1\vb{a}1}\otimes P_0$. Suppose $U$ has Schmidt decomposition
\begin{align}
    U&=\sum_{k=1}^r A_k\otimes B_k\,, \\
    \implies P_{\text{out}}UP_{\text{in}}&=\sum_{k=1}^r P_{1\vb{a}1}A_kP_0\otimes P_0B_kP_{1\vb{a}1}\,,
\end{align}
which implies that $d^{N/2-2}=\mathrm{SR}(P_{\text{out}}UP_{\text{in}})\leq r=\mathrm{SR}(U)$.

This counterexample gives an exact lower bound on the Schmidt rank of the global ancilla unitary connecting two purifications for any finite $N$. Unfortunately, the result trivializes in the thermodynamic limit, because we ask for $\abs{\epsilon}\neq 1$ and the states become trivial, $\sigma^{(N)}(A)\propto \id$ or $\sigma^{(N)}(A)\propto S$. 

At the special point $\abs{\epsilon}=1$, invertibility of $\sigma^{(N)}(A)$ is not guaranteed, and the argument presented above fails. However, the purifications are simple enough to notice that the unitary
\begin{equation}
    U=\epsilon^N S
\end{equation}
is always a solution to $\sigma^{(N)}(A)=\sigma^{(N)}(B)U$, thus $U$ has an MPU representation. The subtlety is that when $\sigma^{(N)}(A)$ is not invertible, $U$ is not the canonical partial isometry $U_0$ of lemma~\ref{lem:purification_partial_isometry}, but rather $U_0=U\, \mathrm{Proj}_{\ker(\sigma^{(N)}(A))^\perp}$. The kernel of $\sigma^{(N)}(A)$ is precisely the eigenspace of $S$ with eigenvalue $\lambda=-1/\epsilon^N$. One verifies that the projector onto any $\lambda$-eigenspace of $S$ is given by
\begin{equation}
    P_\lambda=\frac{1}{N}\sum_{k=0}^{N-1}\lambda^{-k}S^k\,.
\end{equation}
Then $U_0=\epsilon^N S(\id-P_\lambda)=\epsilon^N S+P_\lambda$. A parallel argument to the $\abs{\epsilon}\neq 1$ case shows that $U_0$ has exponential Schmidt rank across the half-chain bipartition, thereby showing that $U_0$ cannot be represented as a tensor network. This illustrative example shows that singular purifications require more care: while the canonical partial isometry $U_0$ can have exponential Schmidt rank, there could exist a unitary extension $U$, still satisfying $A=BU$, which is a unitary relating purifications that is an MPU.

We stress that this counterexample does not preclude the existence of a fundamental theorem. Rather, it is intended to serve as a pedagogical example illustrating two points i) without further reduction procedure, a fundamental theorem is simply not possible, akin to Jordan-block obstructions in MPS, ii) additional subtleties arise when purifications are not full-rank.

\section{Implications on mixed-state SPT phases}
\label{sec:SPT}
In this section, we discuss the implications of the sLPDO fundamental theorem on mixed-state symmetry-protected topological (SPT) phases. In this setting, we consider a many-body mixed state $\rho^{(N)}$ that is invariant under the action of an on-site symmetry $(U_{\mathrm{p}})^{\otimes N}$. Such symmetry is called a weak symmetry if $\rho^{(N)}$ is invariant in the sense that~\cite{SPT_mixed,ma2023average}
\begin{equation}
    U_{\mathrm{p}}^{\otimes N} \rho^{(N)} (U_{\mathrm{p}}^\dagger)^{\otimes N} = \rho^{(N)}.
\end{equation}
In the literature~\cite{SPT_mixed,ma2023average}, the assumption was made that when acting on the purification $\sigma^{(N)}$, the symmetry representation of the purified degrees of freedom is on-site, in the sense that
\begin{equation}
     U_{\mathrm{p}}^{\otimes N} \sigma^{(N)} = \sigma^{(N)}U_{\mathrm{a}}^{\otimes N}. 
\end{equation}
With this assumption, one can work with the vectorized tensor of $A$ and identify non-trivial SPT phases with the nontrivial 2-cocycle. However, as our main result theorem \ref{thm:sLPDO-FT-1} shows, it could be that the invariance of the state $\rho^{(N)}$ under a $U_{\mathrm{p}}^{\otimes N}$ implies that an MPU acts on the purification dimensions instead of an on-site $U_{\mathrm{a}}^{\otimes N}$. We provide examples of mixed states that satisfy the more general intertwining relation
\begin{equation}
\label{eq:MPU_intertwiner}
    U_{\mathrm{p}}^{(N)}\sigma^{(N)}=\sigma^{(N)}U_{\mathrm{a}}^{(N)}\,,    
\end{equation}
where both $U_{\mathrm{p}}^{(N)}$ and $U_{\mathrm{a}}^{(N)}$ are MPUs. This gives states which could potentially exhibit a non-trivial mixed phase of matter.

\subsection{Pulling-through symmetries of sLPDOs}
Homogeneous solutions to Eq.~\eqref{eq:MPU_intertwiner} with $U_{\mathrm{p}}^{(N)}=U_{\mathrm{a}}^{(N)}$, generated by a purification tensor $A$ with $d=D=p$ and $\sqrt{\omega}=\id$ satisfy
\begin{equation}
\label{eq:MPO_pulling_through}
    \begin{tikzpicture}[baseline={([yshift=-0.5ex]current bounding box.center)},x=\cell cm,y=\cell cm]
            \draw (1,0) -- (-1.25,0) -- (-1.25,2);
            \node at (1.5,0) {$\cdots$};
            \draw (2,0) -- (5.5,0) -- (5.5,-1);
            \draw (0,0) -- (0,2);
            \draw (0,-1) -- (0,0);
            \draw (4.5,0) -- (4.5,2);
            \draw (4.5,-1) -- (4.5,0);
            \draw (3,0) -- (3,2);
            \draw (3,-1) -- (3,0);
            \draw (-2, 1.25) -- (1, 1.25);
            \draw (2, 1.25) -- (5.25, 1.25);
            \node at (1.5,1.25) {$\cdots$};
            \MPS{0,0}{1}{blue!20}{\scriptsize $A$}
            \MPS{3,0}{1}{blue!20}{\scriptsize $A$}
            \MPS{4.5,0}{1}{blue!20}{\scriptsize $A$}
            \MPU{-1.25,1.25}{1}{yellow!20}{}
            \MPU{0,1.25}{1}{yellow!20}{}
            \MPU{3,1.25}{1}{yellow!20}{}
            \MPU{4.5,1.25}{1}{yellow!20}{}
            \draw (-2,1.25) to [out=180,in=180] (-2,0.75);
            \draw (5.25,1.25) to [out=0,in=0] (5.25,0.75);
        \end{tikzpicture}=\begin{tikzpicture}[baseline={([yshift=-0.5ex]current bounding box.center)},x=\cell cm,y=\cell cm]
            \draw (1,0) -- (-1,0) -- (-1,1);
            \node at (1.5,0) {$\cdots$};
            \draw (2,0) -- (5.75,0) -- (5.75,-2);
            \draw (0,0) -- (0,1);
            \draw (0,-2) -- (0,0);
            \draw (4.5,0) -- (4.5,1);
            \draw (4.5,-2) -- (4.5,0);
            \draw (3,0) -- (3,1);
            \draw (3,-2) -- (3,0);
            \draw (-0.75, -1.25) -- (1, -1.25);
            \draw (2, -1.25) -- (6.5, -1.25);
            \node at (1.5,-1.25) {$\cdots$};
            \MPS{0,0}{1}{blue!20}{\scriptsize $A$}
            \MPS{3,0}{1}{blue!20}{\scriptsize $A$}
            \MPS{4.5,0}{1}{blue!20}{\scriptsize $A$} 
            \MPU{5.75,-1.25}{1}{yellow!20}{}
            \MPU{0,-1.25}{1}{yellow!20}{}
            \MPU{3,-1.25}{1}{yellow!20}{}
            \MPU{4.5,-1.25}{1}{yellow!20}{}
            \draw (-0.75,-1.25) to [out=180,in=180] (-0.75,-1.75);
            \draw (6.5,-1.25) to [out=0,in=0] (6.5,-1.75);
        \end{tikzpicture}\,.
\end{equation}
An elegant class of solutions comes from tensors and MPUs satisfying the so-called pulling-through condition~\cite{pulling_through}:
\begin{equation}
    \begin{tikzpicture}[baseline={([yshift=-0.5ex]current bounding box.center)},x=\cell cm,y=\cell cm]
            \draw (1,0) -- (-2,0);
            \draw (0,-1) -- (0,2);
            \draw (0.75,2) -- (-2,-0.75);
            \MPS{0,0}{1}{blue!20}{\scriptsize $A$}
            \MPU{0,1.25}{1}{yellow!20}{}
            \MPU{-1.25,0}{1}{yellow!20}{}
        \end{tikzpicture}=\begin{tikzpicture}[baseline={([yshift=-0.5ex]current bounding box.center)},x=\cell cm,y=\cell cm]
            \draw (2,0) -- (-1,0);
            \draw (0,-2) -- (0,1);
            \draw (-0.75,-2) -- (2,0.75);
            \MPS{0,0}{1}{blue!20}{\scriptsize $A$}
            \MPU{0,-1.25}{1}{yellow!20}{}
            \MPU{1.25,0}{1}{yellow!20}{}
        \end{tikzpicture}\,,
\end{equation}
which was originally introduced as an algebraic identity for tensors that describe topological order and SPT. Vectorizing Eq.~\eqref{eq:MPO_pulling_through} and viewing it as an MPS equation, we see that the operator $U_{\mathrm{p}}^{(N)}\otimes (U_{\mathrm{p}}^{(N)})^*$ is a symmetry of the purification. We expect that the relative anomaly class for these operators to vanish, so one cannot claim that these states give rise to interesting phases with certainty. In the next two examples with $U_{\mathrm{p}}^{(N)}\neq U_{\mathrm{a}}^{(N)}$, the relative anomaly classes are non-trivial and could give rise to states in a non-trivial phase.

\subsection{sLPDO with weak $\mathbb{Z}_2$ symmetry}
\label{sec:Z_2_sLPDO}
We construct a homogeneous step-injective sLPDO with an on-site $U_{\mathrm{p}}^{(N)}=Z^{\otimes N}$ physical weak symmetry and an MPU $U_{\mathrm{a}}^{(N)}=(-1)^{N/2}CZY^{(N)}$ in the purification indices, where $CZY^{(N)}=\left(\prod_{i=1}^N CZ_{i,i+1}\right) Y^{\otimes N}$ is known to be anomalous. The example is based on a `bare' structure which exhibits the requirements of an intertwining relation, but its density operator is a commuting product of nearest-neighbor terms, which in fact is a thermal state of the Ising chain, which is considered in the trivial phase. In appendix~\ref{app:Z_2_sLPDO}, we show that we can move away from this state by `dressing' the bare state with an MPO commuting with the on-site $Z^{\otimes N}$ symmetry to obtain a family of possibly non-trivial states. Appendix~\ref{app:Z_2_sLPDO} also contains further details of the bare state. In the following, all operations are in $\mathbb{F}_2$, performed with binary addition and binary multiplication.

We start by constructing the purification tensor $A^{\{2\}}$ of the bare state. We proceed first by introducing a purification tensor $T$ which has $d=p=2$ and two memory qubits $D=4$ from which we will construct $A$, and $A^{\{2\}}$ is obtained by blocking every two sites. We introduce the tensor
\begin{align}
    T^{ia}_{(c',y)(c,x)}&=\lambda_{xa}(-1)^{c(xa+a)}\delta_{ya}\delta_{c'c}\delta_{i,x} \nonumber \\
    &=\lambda_{xa}(-1)^{c(xa+a)}\ \begin{tikzpicture}[baseline={([yshift=-0.5ex]current bounding box.center)},x=\cell cm,y=\cell cm]
            \begin{knot}[clip width=3]
                \strand (-1,0.1) -- (1,0.1);
                \strand (1,-0.1) .. controls +(-1,0) and +(0,-1) .. (0,1);
            \end{knot}
            \draw (-1,-0.1) .. controls +(1,0) and +(0,1) .. (0,-1);
            \node at (0,1.25) {\tiny $i$};
            \node at (0,-1.25) {\tiny $a$};
            \node at (1.25,0.1) {\tiny $c$};
            \node at (1.25,-0.1) {\tiny $x$};
            \node at (-1.25,0.2) {\tiny $c'$};
            \node at (-1.25,-0.1) {\tiny $y$};
        \end{tikzpicture}\,,
\end{align}
where
\begin{equation}
    \lambda_{xa}=\begin{cases}
        \sqrt{\frac{1+\mu}{2}}\,, & a=x \\
        \sqrt{\frac{1-\mu}{2}}\,, & a\neq x
    \end{cases}
\end{equation}
for $0<\abs{\mu}<1$ is some non-zero weighting. The boundary is chosen to be $(\sqrt{\tau})_{(c',y),(c,a_0)}=\frac{1}{\sqrt{2}}\lambda_{ca_0}(-1)^c \delta_{c'c}\delta_{y,a_0}$. 

Because the final memory is a physical site in the sLPDO ansatz, it is convenient to block every two sites to form the tensor $T^{\{2\}}$, which now has $d=p=D=4$. 

The tensor $A$ is obtained from $T$ by conjugating the two virtual qubits by a $CZ$ gate, and then performing a Hadamard rotation on all qubits except the purification qubits. That is,
\begin{align}
    A&=((H\otimes H)_{\mathrm{p}}\otimes F_{\mathrm{virtual}})T(\id_{\mathrm{a}}\otimes F_{\mathrm{virtual}}^\dagger)\,, \\
    F_{\mathrm{virtual}}&=(H\otimes H)CZ\,.
\end{align}
Then, $A^{\{2\}}$ is obtained by blocking every two sites of $A$. In components, this reads
\begin{align}
\label{eq:Z_2_sLPDO_purification_tensor}
    (A^{\{2\}})^{(i_1,i_2)(a_1,a_2)}_{(c',y)(c,x)}&=\lambda_{ra_1}\lambda_{a_1a_2}(-1)^{\gamma(ra_1+a_1+a_1a_2+a_2)} \nonumber \\
    \phantom{={}}&\times\begin{tikzpicture}[baseline={([yshift=-0.5ex]current bounding box.center)},x=\cell cm,y=\cell cm]
            \begin{knot}[clip width=3, clip radius=3pt, end tolerance=1pt]
                \strand (-3.25,0.3) -- (1.75,0.3);
                \strand (1,-0.3) .. controls +(-1,0) and +(0,-1) .. (0,1);
                \strand (-0.75,-0.3) .. controls +(-0.8,0) and +(0,-1) .. (-1.5,1);
            \end{knot}
            \draw (-0.75,-0.3) .. controls +(0.75,0) and +(0,0.8) .. (0,-1);
            \draw (-2.5,-0.3) .. controls +(1,0) and +(0,0.8) .. (-1.5,-1);
            \draw (-3.25,-0.3) -- (-2.5,-0.3);
            \draw (1.75,-0.3) -- (1,-0.3);
            \MPS{0,0.75}{0.5}{white}{}
            \MPS{-1.5,0.75}{0.5}{white}{}
            \MPS{1.25,0.3}{0.5}{white}{}
            \MPS{-2.75,0.3}{0.5}{white}{}
            \MPS{1.25,-0.3}{0.5}{white}{}
            \MPS{-2.75,-0.3}{0.5}{white}{}
            \draw (0,1) -- (0,1.25);
            \draw (-1.5,1) -- (-1.5,1.25);
            \filldraw (0.75,0.3) circle (1.5pt);
            \filldraw (0.75,-0.3) circle (1.5pt);
            \filldraw (-2.25,0.3) circle (1.5pt);
            \filldraw (-2.25,-0.3) circle (1.5pt);
            \draw (0.75,0.3) -- (0.75,-0.3);
            \draw (-2.25,0.3) -- (-2.25,-0.3);
            \node at (-0.75, 0.45) {\tiny $\gamma$};
            \node at (0.5, -0.1) {\tiny $r$};
            \node at (0,1.5) {\tiny $i_1$};
            \node at (-1.5,1.5) {\tiny $i_2$};
            \node at (2,0.3) {\tiny $c$};
            \node at (2,-0.3) {\tiny $x$};
            \node at (-3.5,0.3) {\tiny $c'$};
            \node at (-3.5,-0.3) {\tiny $y$};
            \node at (0,-1.25) {\tiny $a_1$};
            \node at (-1.5,-1.25) {\tiny $a_2$};
        \end{tikzpicture}\,,
\end{align}
where we labeled also the internal summed $\gamma$ and $r$ indices and
\begin{equation}
    H=\begin{tikzpicture}[baseline={([yshift=-0.5ex]current bounding box.center)},x=\cell cm,y=\cell cm]
            \draw (-0.5,0) -- (0.5,0);
            \MPS{0,0}{0.5}{white}{}
        \end{tikzpicture}, \qquad CZ_{i,i+1}=\begin{tikzpicture}[baseline={([yshift=-0.5ex]current bounding box.center)},x=\cell cm,y=\cell cm]
            \filldraw (0,0.3) circle (1.5pt);
            \filldraw (0,-0.3) circle (1.5pt);
            \draw (0,0.3) -- (0,-0.3);
            \draw (-0.5,0.3) -- (0.5,0.3);
            \draw (-0.5,-0.3) -- (0.5,-0.3);
        \end{tikzpicture}\,.
\end{equation}
We will also do the same to the boundary:
\begin{equation}
    \sqrt{\omega}=(H\otimes H)CZ\sqrt{\tau}\,.
\end{equation}

The tensor $A^{\{2\}}$ is step-injective: i) as a map, $T:\mathbb{C}^p\otimes \mathbb{C}^D\to \mathbb{C}^d\otimes \mathbb{C}^D$ is a weighted permutation; ii) all weights are non-zero and so $T$ and therefore $T^{\{2\}}$ is invertible; iii) $A^{\{2\}}$ is unitarily related to $T^{\{2\}}$ and is also invertible. Furthermore, $\sqrt{\omega}$ is full-rank. Therefore, when this representation is compared with any other sLPDO representation generating the same density operators for all system sizes, the step-injective fundamental theorem~\ref{thm:sLPDO-FT-1} applies and the corresponding purification unitary admits an MPI representation.

We fix a size convention before proceeding. Let $M$ be the number of $A^{\{2\}}$ tensors in the representation, so $\sigma^{(M)}(A,\sqrt{\omega})$ is a purification on $M+1$ sites with local Hilbert space dimension $4$. We may, by an abuse of notation, interpret $\sigma^{(M)}$ also as a purification on $2M+2$ sites with local Hilbert space dimension $2$. Let $i_k$ be the physical index carrying the purification index $a_k$.

We claim that the purification tensor $\sigma^{(M)}(A,\sqrt{\omega})$ satisfies 
\begin{equation}
    Z^{\otimes(2M+2)}\sigma^{(M)}=(-1)^{M+1}\sigma^{(M)}CZY^{(2M+2)}\,,
\end{equation}
whose computational details are left to appendix~\ref{app:Z_2_sLPDO}. The corresponding density matrix, written in local Hilbert space dimension $d=2$, is a commuting nearest-neighbor product
\begin{align}
\label{eq:Z_2_sLPDO_state}
    \rho^{(M)}&=\sigma^{(M)}{\sigma^{(M)}}^\dagger \nonumber \\
    &=\frac{1}{2^{2M+2}}\left(\id+\mu X_c X_{i_0}\right)\prod_{k=1}^{2M}\left(\id+\mu X_{i_{k-1}}X_{i_k}\right)\,.
\end{align}
In fact, $\rho^{(M)}$ is a thermal state of the classical Ising chain in the $X$-basis,
\begin{align}
    \rho^{(M)}&=\frac{1}{\mathcal{Z}}\exp[-\arctanh(\mu) H_{\mathrm{Ising}}]\,, \nonumber \\
    H_{\mathrm{Ising}}&=-X_cX_{i_0}-\sum_{k=1}^{2M}X_{i_{k-1}}X_{i_k}\,,
\end{align}
which clearly has the $\mathbb{Z}_2$ weak symmetry $Z^{\otimes(2M+2)}$. For any non-zero $\mu$, correlation functions decay exponentially as $\langle X_{i_j}X_{i_k}\rangle_{\rho^{(M)}}=\mu^{\abs{j-k}}$, and the $\mu\to 0$ limit corresponds to the infinite-temperature maximally mixed state limit. This completes the construction of the bare state $\rho^{(M)}$.

Although the bare purification already realizes the desired anomalous purification-space symmetry, its density matrix belongs to a commuting $X$-diagonal family. It is possible to deform it by acting with an invertible physical MPO commuting with $Z^{\otimes(2M+2)}$, thereby preserving both step-injectivity and also the purification MPU, while taking the density operator outside the commuting $X$-diagonal family. We simply state the existence of such an MPO in this section, but provide an explicit example in appendix~\ref{app:Z_2_sLPDO_MPO}. The MPO makes the resulting dressed state long-range correlated, which already exhibits different properties from the bare state.

\subsection{Uniform LPDO with weak $\mathbb{Z}_2$ symmetry}
For completeness, we also present an example from the PBC uniform LPDO ansatz. Even when $\rho^{(N)}(A)$ possesses an on-site weak symmetry, when the symmetry representation on the purified degrees of freedom is an MPU,
\begin{align}
    \label{eq:sym}
        &\begin{tikzpicture}[baseline={([yshift=-0.5ex]current bounding box.center)},x=\cell cm,y=\cell cm]
            \draw (2.5,0) -- (-1,0);
            \node at (3,0) {$\cdots$};
            \draw (3.5,0) -- (5.5,0);
            \draw (0,0) -- (0,2);
            \draw[dashed] (0,-1) -- (0,0);
            \draw (4.5,0) -- (4.5,2);
            \draw[dashed] (4.5,-1) -- (4.5,0);
            \draw (1.5,0) -- (1.5,2);
            \draw[dashed] (1.5,-1) -- (1.5,0);
            \MPS{0,0}{1}{blue!20}{\scriptsize $A$}
            \MPS{1.5,0}{1}{blue!20}{\scriptsize $A$}
            \MPS{4.5,0}{1}{blue!20}{\scriptsize $A$}
            \MPU{0,1.25}{1}{yellow!20}{\scriptsize $U_{\mathrm{p}}$}
            \MPU{1.5,1.25}{1}{yellow!20}{\scriptsize $U_{\mathrm{p}}$}
            \MPU{4.5,1.25}{1}{yellow!20}{\scriptsize $U_{\mathrm{p}}$}
        \end{tikzpicture}=\begin{tikzpicture}[baseline={([yshift=-0.5ex]current bounding box.center)},x=\cell cm,y=\cell cm]
            \draw (2.5,0) -- (-1,0);
            \node at (3,0) {$\cdots$};
            \draw (3.5,0) -- (5.5,0);
            \draw (0,0) -- (0,1);
            \draw[dashed] (0,-2) -- (0,0);
            \draw (4.5,0) -- (4.5,1);
            \draw[dashed] (4.5,-2) -- (4.5,0);
            \draw (1.5,0) -- (1.5,1);
            \draw[dashed] (1.5,-2) -- (1.5,0);
            \MPS{0,0}{1}{blue!20}{\scriptsize $A$}
            \MPS{1.5,0}{1}{blue!20}{\scriptsize $A$}
            \MPS{4.5,0}{1}{blue!20}{\scriptsize $A$}
            \draw (2.5,-1.25) -- (-1,-1.25);
            \node at (3,-1.25) {$\cdots$};
            \draw (3.5,-1.25) -- (5.5,-1.25);
            \MPU{0,-1.25}{1}{purple!20}{\scriptsize $U_{\mathrm{a}}$}
            \MPU{1.5,-1.25}{1}{purple!20}{\scriptsize $U_{\mathrm{a}}$}
            \MPU{4.5,-1.25}{1}{purple!20}{\scriptsize $U_{\mathrm{a}}$}
        \end{tikzpicture}\, ,
    \end{align}
a direct result is that if such an MPU $U_{\mathrm{a}}^{(N)}$ is anomalous, then the tensor $A$ that generates $\sigma^{(N)}(A)$ is non-injective\footnote{As a map from two horizontal bonds to two vertical bonds, which is different from the definition of step-injective. }~\cite{GarreRubio2023classifyingphases}, and therefore the tensor generating $\rho^{(N)}(A)$ is non-injective. This brings the possibility that $\rho^{(N)}(A)$ cannot be two-way connected to the trivial state, and may bring new mixed-state phases.

As a simple example, consider $\sigma^{(N)}\propto \id^{\otimes N} + CZX^{(N)}$ for even $N$, where $CZX^{(N)}=\prod_{i=1}^N (CZ)_{i,i+1} \prod_{i=1}^N X_i$. The density matrix $\rho^{(N)}=\sigma^{(N)} (\sigma^{(N)})^\dagger$ allows an LPDO representation. The purified state $\sigma^{(N)}$ has the following symmetry,
\begin{equation}
    Z^{\otimes N} \sigma^{(N)}=\sigma^{(N)} CZY^{(N)},
\end{equation}
just as in the last example. This clearly shows that in the case where $\rho^{(N)}$ has only a weak $\mathbb{Z}_2$ symmetry $Z^{\otimes N}$, it can still be in the non-trivial phase~\cite{lessa2025mixed,sun2025anomalous,liu2025trading}, although it is not associated with a non-trivial 2-cocycle. Denote the three basis of virtual space as $\{|I\rangle,|U,k\rangle_{k=0,1}\}$. The tensor $A$ generating $\sigma^{(N)}$ (under PBC) takes the form of
\begin{equation}
    A^{ia}=\frac{1}{\sqrt{2}}\left[\delta_{i,a}|I\rangle\langle I|+\delta_{i,\bar{a}}\sum_k (-1)^{ik}|U,i\rangle\langle U, k|\right],
\end{equation}
and one can formally write down an sLPDO representation using the procedure in Sec.~\ref{sec:def-sLPDO}. We note that the resulting tensor is not step-injective. 

One may question that in the above example, $\rho^{(N)}$ has a strong symmetry $CZX^{(N)}\rho^{(N)}=\rho^{(N)}$. To break the strong symmetry, consider the following purified state 
\begin{equation}
    \sigma^{(N)}\propto M^{\otimes N}+ Z^{\otimes N} M^{\otimes N} CZY^{(N)},\quad M = \begin{pmatrix}
        1 & \alpha\\
        \alpha & 1
    \end{pmatrix},
\end{equation}
with $\alpha>0$, which has the same weak symmetry by construction. The corresponding mixed state $\rho^{(N)}$ is
\begin{equation}
\begin{aligned}
    \rho^{(N)}&=\sigma^{(N)}(\sigma^{(N)})^\dag \\
    &\propto (M^2)^{\otimes N}+(\tilde{M}^2)^{\otimes N}\\
    &\quad +M^{\otimes N} CZX^{(N)} \tilde{M}^{\otimes N} + \tilde{M}^{\otimes N}CZX^{(N)}M^{\otimes N}
\end{aligned}
\end{equation}
where $\tilde{M}=Z M Z$. $\rho^{(N)}$ has long-range order in that its two-point correlator is $\langle X_i X_j\rangle=\frac{4\alpha^2}{(1+\alpha^2)^2}$ for any $i\neq j$ and $\langle X_i\rangle=0$. Still, one can ask whether it is in the same phase as the classical GHZ state, which we leave as an interesting open question. 

\section{Discussion and outlook}
\label{sec:conclusion}
We have taken the first steps toward understanding the structure of mixed-state tensor networks in the form of LPDOs. We introduced the sLPDO ansatz as a class of states generated sequentially by applying successive channels onto an initial state. This class contains several states of interest, including boundaries of topologically-ordered systems and certain thermal states of local Hamiltonians. 

We identified two independent sufficient conditions under which equivalent sLPDO representations can be locally related, thereby establishing fundamental theorems for sLPDOs. A key observation motivating this analysis is that two purifications of the same state are necessarily related by a global partial isometry. We studied the compatibility between this global partial isometry and the tensor network structure of the state. First, we showed that step-injectivity of the tensor allows the purification to have a tensor network left inverse. This implied that the isometry relating purifications is also itself a tensor network, thus relating the local tensors $A$ and $B$ through an MPI. Second, we identified a cyclicity condition in the virtual memory space, which is a property of the state itself and independent of the representation. The cyclic condition promises that two equivalent representations are related by an on-site isometry, which interestingly is substantially stronger than the conclusion obtained from step-injectivity.

While these results shine light on what a fundamental theorem for uniform LPDOs might involve, we showed via an example that a na\"ive extension to uniform LPDOs fails. We considered two purifications generating the same state, and proved that the unitary relating purifications has Schmidt rank growing exponentially with system size. We also raised a subtlety that arises when the purifications are not full rank: it could be that there exists a unitary connecting purifications that is an MPU, but it might not be the unique canonical one. 

Finally, we commented on what a fundamental theorem could imply for the phase classification of mixed states. A possibility missed in previous literature, the MPU connecting purifications could carry a non-trivial cocycle. We expect a systematic study of these cases to provide new, interesting classes of quantum phases. 

A complete fundamental theorem for uniform LPDOs remains an important open question. Such a theorem would require a systematic understanding of the LPDO ansatz, including the interplay between horizontal virtual and vertical purification degrees of freedom. It would also have to be compatible with the existing canonical forms and fundamental theorems for MPVs. It is unclear what the steps to tackle this question should look like. While a general route toward such a theorem is not yet apparent, the case-by-case analysis developed here identifies representative structures and obstructions that any eventual LPDO fundamental theorem must account for.

\begin{acknowledgments}
Y.Y.\@ thanks M.\@ Florido-Llinàs and G.\@ Styliaris for helpful discussions. Y.Y.\@ is grateful to E.\@ Tjoa for the suggestion to consider convex combinations of MPSs. 

Y.Y.\@ acknowledges support from the International Max Planck Research School for Quantum Science and Technology (IMPRS-QST). Y.L.\@ is supported by the Alexander von Humboldt Foundation. This research is part of the Munich Quantum Valley (MQV), which is supported by the Bavarian state government with funds from the Hightech Agenda Bayern Plus. This work has also been partially supported by the Klaus Tschira Foundation.
\end{acknowledgments}

\appendix

\section{Proof of lemma~\ref{lem:purification_partial_isometry}}
\label{app:proof}
\begin{proof}
    The reverse direction is trivial. For the forward direction, write the polar decompositions $A=\sqrt{AA^\dagger} V_A$ and $B=\sqrt{BB^\dagger} V_B$, where $V_A V_A^\dagger=\mathrm{Proj}_{\mathrm{col}(A)}$, $V_A^\dagger V_A = \mathrm{Proj}_{\mathrm{ker}(A)^\perp}$, and similarly for $V_B$. The required partial isometry is $U=V_B^\dagger V_A\in\mathcal{M}_{p_B,p_A}$, since 
    \begin{align}
        BU&=\sqrt{BB^\dagger}V_B V_B^\dagger V_A \nonumber \\
        &=\sqrt{BB^\dagger}\mathrm{Proj}_{\mathrm{col}(B)}V_A \\
        &=\sqrt{AA^\dagger}V_A=A\,.
    \end{align} 
    The partial isometry also satisfies $U^\dagger U=V_A^\dagger \mathrm{Proj}_{\mathrm{col}(B)}V_A=V_A^\dagger \mathrm{Proj}_{\mathrm{col}(A)}V_A=V_A^\dagger V_A=\mathrm{Proj}_{\mathrm{ker}(A)^\perp}$ and $UU^\dagger=V_B^\dagger V_AV_A^\dagger V_B=V_B^\dagger \mathrm{Proj}_{\mathrm{col}(A)}V_B=V_B^\dagger V_B=\mathrm{Proj}_{\mathrm{ker}(B)^\perp}$.

    For uniqueness, suppose $V_0$ is another partial isometry satisfying the properties Eq.~\eqref{eq:partial_isometry_properties}. Let $x\in \mathbb{C}^{p_A}$ and we aim to show $U_0x=V_0x$. By subtracting equations, we have that $U_0x-V_0x\in\ker B$. However, $U_0U_0^\dagger$ is the projector onto the column space of $U_0$, showing that the range of $U_0$ is $\ker(B)^\perp$, similarly for $V_0$. Clearly $\ker(B)\cap\ker(B)^\perp=\{0\}$, so $U_0x-V_0x=0$ which completes the argument.
    \end{proof}

\section{Fundamental theorem for short-range entangled mixed states}
\label{app:CCDO}
In this appendix, we consider the freedom in representation of density matrices which are generated by ensembles $\{p_k,\ket{\Psi^{(N)}(A_k)}\}_{k=1}^K$ of normal MPS, where $\ket{\Psi^{(N)}(A_k)}$ is the uniform MPS generated by a normal tensor $A_k\in\mathcal{M}_{D_k}$. That is, the density matrix is 
\begin{align}
\label{eq:convex_MPS}
    \rho^{(N)}(\{p_k,A_k\})&=\sum_{k=1}^K p_k\ketbra{\Psi^{(N)}(A_k)}{\Psi^{(N)}(A_k)} \nonumber \\
    &=\sum_{k=1}^K\ketbra{\Psi^{(N)}(A_k;\sqrt{P_k})}{\Psi^{(N)}(A_k;\sqrt{P_k})} \nonumber \\
    &=\sum_{k=1}^K\begin{tikzpicture}[baseline={([yshift=-0.5ex]current bounding box.center)},x=\cell cm,y=\cell cm]
            \draw (-1, -0.5) to [out=180,in=180] (-1,0) -- (2.5,0);
            \node at (3,0) {$\cdots$};
            \draw (3.5,0) -- (5.5,0) to [out=0,in=0] (5.5,-0.5);
            \draw (4.5,0) -- (4.5,1);
            \draw (1.5,0) -- (1.5,1);
            \MPU{0,0}{1}{yellow!20}{\scriptsize $\sqrt{P_k}$}
            \MPS{1.5,0}{1}{blue!20}{\scriptsize $A_k$}
            \MPS{4.5,0}{1}{blue!20}{\scriptsize $A_k$}
            \begin{scope}[shift={(0,-1.5)}]
                \draw (-1, -0.5) to [out=180,in=180] (-1,0) -- (2.5,0);
                \node at (3,0) {$\cdots$};
                \draw (3.5,0) -- (5.5,0) to [out=0,in=0] (5.5,-0.5);
                \draw (4.5,-1) -- (4.5,0);
                \draw (1.5,-1) -- (1.5,0);
                \MPU{0,0}{1}{yellow!20}{\scriptsize $\sqrt{P_k}$}
                \MPS{1.5,0}{1}{blue!20}{\scriptsize $\overline{A_k}$}
                \MPS{4.5,0}{1}{blue!20}{\scriptsize $\overline{A_k}$}
            \end{scope}
        \end{tikzpicture}\,,
\end{align}
where $P_k=p_k\id_{D_k}$. We make the additional mild assumption that the number of labels $K$ does not grow with system size, $K=O(1)$. This is a reasonable assumption when the tensors are homogeneous.

When each MPS tensor $A_k$ in the ensemble is normal, states of the form Eq.~\eqref{eq:convex_MPS} can be understood as mixtures of short-range entangled (SRE) states, leading to a working definition of SRE mixed states and complexity \cite{Mixed_state_complexity}. Previous definitions of SRE mixed states only included SRE pure states preparable by a constant depth unitary circuit from a product state \cite{TO_nonzero_T,Local_decoherence_separability,Symmetry_separability}, but recall that gapped SRE entangled states are only faithfully described by normal MPSs, which in general, optimally require circuit depth $\Omega(\log(N))$ to prepare \cite{MPS_prep}.

Every state Eq.~\eqref{eq:convex_MPS} can be represented trivially as a PBC LPDO with diagonal boundary. It relies on the fact that an on-site purification of $\rho^{(N)}(\{p_k,A_k\})$ exists and is $\sum_k\ket{\Psi^{(N)}(A_k;\sqrt{P_k})}\otimes\ket{k}^{\otimes N}$, which has what we call a GHZ purification space ${\ket{k}^{\otimes N}}$. The purification tensor $L^{ia}=\bigoplus_{k=1}^K A^i_k\delta_{ak}$ with boundary $\sqrt{Q}=\bigoplus_{k=1}^K\sqrt{P_k}$ generates $\rho^{(N)}(\{p_k,A_k\})$. To be explicit,
\begin{equation}
    \rho^{(N)}(\{p_k,A_k\})=\begin{tikzpicture}[baseline={([yshift=-0.5ex]current bounding box.center)},x=\cell cm,y=\cell cm]
            \draw (-1, -0.5) to [out=180,in=180] (-1,0) -- (2.5,0);
            \node at (3,0) {$\cdots$};
            \draw (3.5,0) -- (5.5,0) to [out=0,in=0] (5.5,-0.5);
            \draw (4.5,0) -- (4.5,1);
            \draw (1.5,0) -- (1.5,1);
            \draw[dashed] (1.5,0) -- (1.5,-1);
            \draw[dashed] (4.5,0) -- (4.5,-1);
            \MPU{0,0}{1}{yellow!20}{\scriptsize $\sqrt{Q}$}
            \MPS{1.5,0}{1}{blue!20}{\scriptsize $L$}
            \MPS{4.5,0}{1}{blue!20}{\scriptsize $L$}
            \begin{scope}[shift={(0,-1.5)}]
                \draw (-1, -0.5) to [out=180,in=180] (-1,0) -- (2.5,0);
                \node at (3,0) {$\cdots$};
                \draw (3.5,0) -- (5.5,0) to [out=0,in=0] (5.5,-0.5);
                \draw (4.5,-1) -- (4.5,0);
                \draw (1.5,-1) -- (1.5,0);
                \MPU{0,0}{1}{yellow!20}{\scriptsize $\sqrt{Q}$}
                \MPS{1.5,0}{1}{blue!20}{\scriptsize $\overline{L}$}
                \MPS{4.5,0}{1}{blue!20}{\scriptsize $\overline{L}$}
            \end{scope}
        \end{tikzpicture}
\end{equation}
Thus we view convex combinations of MPSs as a subclass of LPDOs.

Given the equality of LPDOs 
\begin{equation}
\label{eq:LPDO_equality}
    \rho^{(N)}(\{p_k,A_k\})=\rho^{(N)}(\{q_l,B_l\})\,, \quad \forall N\geq1\,,
\end{equation}
we ask how the ensembles $\{p_k,A_k\}$ and $\{q_l,B_l\}$ are related. We prove, in the following, that the ensembles can be related by an on-site permutation matrix with phases. On the level of the tensor, this implies the tensors are related on the purification index by an on-site unitary.

First, we say that two normal tensors $A$ and $B$ are equivalent, $A\sim B$, if they generate proportional states for all $N$, which was shown to be equivalent to the existence of a phase $\theta$ and invertible matrix $X$ such that $B^i=e^{i\theta} XA^iX^{-1}$ \cite{MPDO_RFP}. Without loss of generality, we can ensure that the $\{A_k\}_{k=1}^K$ ($\{B_l\}_{l=1}^L$) form a basis of normal tensors (BNT), that is, that $\{\ket{\Psi^{(N)}(A_k)}\}_{k=1}^K$ ($\{\ket{\Psi^{(N)}(B_l)}\}_{l=1}^L$) are asymptotically pairwise orthogonal. Thus no two distinct elements of a BNT are equivalent.

Eq.~\eqref{eq:LPDO_equality} implies that there exists a unitary $U^{(N)}$ relating the ensembles via
\begin{equation}
    \sqrt{p_k}\ket{\Psi^{(N)}(A_k)}=\sum_{l}U_{kl}^{(N)}\sqrt{q_l}\ket{\Psi^{{N}}(B_l)}\,.
\end{equation}
We will show that this unitary can be chosen as a tensor product of on-site unitaries.

Choose a label $k$ and compute
\begin{align}
    &\mel{\Psi^{(N)}(A_k)}{\rho^{(N)}}{\Psi^{(N)}(A_k)} \nonumber \\
    &=\sum_{k'=1}^K p_{k'}\times\abs{\braket{\Psi^{(N)}(A_k)}{\Psi^{(N)}(A_{k'})}}^2 \nonumber \\
    &\to p_k
\end{align} 
as $N\to\infty$. But also
\begin{align}
    &\mel{\Psi^{(N)}(A_k)}{\rho^{(N)}}{\Psi^{(N)}(A_k)} \nonumber \\
    &=\sum_{l=1}^L q_{l} \abs{\braket{\Psi^{(N)}(A_k)}{\Psi^{(N)}(B_l)}}^2\,.
\end{align}
From the above, it must be that exactly one of the $B_l$ is equivalent to $A_k$ with $p_k=q_l$. Doing this for each $k$, it cannot be that $B_l\sim A_k$ and $B_l\sim A_{k'}$ since that would imply $A_k\sim A_{k'}$, contradicting the BNT assumption. Therefore, it must be that $K=L$ and furthermore that there exists a permutation $\pi\in S_K$ such that $A_k\sim B_{\pi(k)}$ and $p_k=q_{\pi(k)}$.

In particular, one can choose the unitary $U^{(N)}_{kl}$ relating ensembles to be $U^{(N)}_{kl}=e^{-i\theta_k N}\delta_{l,\pi(k)}$, which is realized on the GHZ purification subspace as an on-site unitary $U^{(N)}=\left(\sum_{k,l}e^{-i\theta_l}\delta_{k,\pi(l)}\ketbra{k}{l}\right)^{\otimes N}$.

We comment briefly on the case where $A_k$ are allowed to be non-normal. A similar approach can be taken, but in general, there is no more control over the unitary $U^{(N)}$ and we expect a similar example to section~\ref{sec:counterexample} to exist in this context.

\section{Boundary of topological orders}
\label{app:boundary}
\subsection{Construction of the boundary state from PEPS}
\begin{proposition}
    Given a $C^*$-weak Hopf algebra $\mathcal{A}$ and two faithful $*$-representation $\phi$ and $\psi$ of $\mathcal{A}$ and $\mathcal{A}^*$. Take the canonical regular element $\omega\in\mathcal{A}^*$ and define a matrix $b(\omega)$ such that $\tr[b(\omega)\phi(x)]=\omega(x)$ for any $x\in\mathcal{A}$. Then, the rank-4 tensor
\begin{equation}
    M_{\alpha\beta}^{ij}=\sum_{a=1}^{\mathrm{dim}(\mathcal{A})} [b(\omega)\phi(e_a)]_{ij}\otimes [\psi(e^a)]_{\alpha\beta}
\end{equation}
generates a mixed state at the  renormalization fixed point. Furthermore, such a mixed state corresponds to the boundary of topological order. 
\end{proposition}

We now apply this result to the case where $\mathcal{A}=\mathbb{C}[G]$ is the group algebra of the group $G$. We identify the algebra basis $\{e_a\}$ with the group elements $\{g_a\}$. The canonical regular element $\omega$ of $\mathcal{A}^*$ and $\Omega$ of $\mathcal{A}$ are
\begin{equation}
    \omega=g^1,\quad \Omega=\frac{1}{|G|}\sum_{g_a\in G} g_a,
\end{equation}
where $g_1$ is the group identity and $g^1$ is the corresponding element in the dual algebra. Next, choose the faithful representations as
\begin{equation}
    \phi=L,\quad [\psi(g^a)]_{\alpha\beta}=\delta_{\alpha=\beta=a},
\end{equation}
where $L$ denotes the regular representation. Both $\psi$ and $\phi$ are then $|G|$-dimensional representations. One can define $b(\Omega)$ as $\tr[b(\Omega)\psi(y)]=\Omega(y)$ for any $y\in\mathcal{A}^*$ and show a valid choice of $b(\Omega)$ is $b(\Omega)=\frac{1}{|G|}\id$. 
By definition of $b(\omega)$, a valid choice is $b(\omega)=\frac{1}{|G|}\phi(g_1)=\frac{1}{|G|}\id$. Therefore, the local MPDO tensor is
\begin{equation}
    M_{\alpha\beta}^{ij} = \frac{1}{|G|} [L(g_\alpha)]_{ij} \delta_{\alpha,\beta}. 
\end{equation}
The mixed state it generates is then
\begin{equation}
    \rho^{(N)}=\frac{1}{|G|^N}\sum_g \sum_{g_1\cdots g_N}|gg_1\cdots gg_N\rangle\langle g_1\cdots g_N|. 
\end{equation}
One can show by taking
\begin{equation}
    \sigma^{(N)}=\frac{1}{\sqrt{|G|}^N}\sum_g \sum_{g_1\cdots g_N}|gg_1\cdots gg_N\rangle\langle g_1\cdots g_N|
\end{equation}
then $\sigma^{(N)}{\sigma^{(N)}}^{\dagger} = \abs{G}\rho^{(N)}$. 

\subsection{Abelian groups have diagonal representation}
\label{app:abelian_boundary_TO}
In this section, we outline the construction of a diagonal density matrix representing the boundary mixed state of $D(G)$ topological order. We denote by $\mathbb{C}[G]$ the group algebra which is the vector space formally spanned by elements of the group $G$. Here we will sometimes explicitly label the space in which vectors live, for example that $\ket{g}_{G}\in\mathbb{C}[G]$.

For abelian $G$, every irreducible representation is one-dimensional and is one-to-one with characters $\chi$, which are group homomorphisms $\chi:G\to U(1)$. The group of characters, denoted by $\widehat{G}$, is isomorphic to the group $G$, and allows to define the group Fourier transform $U:\mathbb{C}[\widehat{G}]\to\mathbb{C}[G]$ given by $U\ket{\chi}_{\widehat{G}}=\frac{1}{\sqrt{\abs{G}}}\sum_{g}\chi(g^{-1})\ket{g}_G$. The upshot is that the multiplication operator $L_h:\mathbb{C}[G]\to\mathbb{C}[G]$ is diagonal in this basis and has action $U^\dagger L_h U\ket{\chi}_{\widehat{G}}=\chi(h)\ket{\chi}_{\widehat{G}}$.

We recall that the boundary state for general $G$ is given by
\begin{equation}
    \rho^{(N)}_G=\frac{1}{\abs{G}^N}\sum_{h\in G}L_h^{\otimes N}\,.
\end{equation}
Applying a site-wise basis change yields
\begin{align}
\label{eq:boundary_TO_diagonal}
    \tilde{\rho}^{(N)}_G&:=(U^\dagger)^{\otimes N}\rho^{(N)}_G U^{\otimes N} \nonumber \\
    &=\frac{1}{\abs{G}^N}\sum_{h\in G}\sum_{\vb*{\chi}\in\widehat{G}^N}(\chi_1\cdots\chi_N)(h)\ketbra{\chi_1\cdots\chi_N}{\chi_1\cdots\chi_N} \nonumber \\
    &=\frac{1}{\abs{G}^{N-1}}\sum_{\substack{\chi_1,\dots,\chi_N\in \widehat{G} \\ \chi_1\cdots \chi_N=1}}\ketbra{\chi_1\cdots \chi_N}{\chi_1\cdots \chi_N}\,,
\end{align}
where we used the character orthogonality relation
\begin{equation}
    \sum_{h\in G}(\chi\eta)(h)=\abs{G}\delta_{\chi\eta,1}\,.
\end{equation}

The $N$-site sLPDO $\rho^{(N)}(A;\omega)$ where $A^{ia}_{\alpha\beta}=\frac{1}{\sqrt{\abs{G}}}\delta^{i,\alpha a}\delta_{\alpha\beta}$ with $i,a,\alpha,\beta\in G$ and $\omega=\ketbra{+_G}{+_G}$ is equal to $\rho^{(N+1)}_G$. To obtain the tensor generating the diagonal state, it is not enough to apply the inverse Fourier transform to the physical index of $A$, because the final virtual degree of freedom is interpreted as a physical site in the sLPDO. In fact, the tensor $B$ generating $\tilde{\rho}^{(N+1)}_G$ is obtained from $A$ by Fourier transforming all four indices\footnote{In principle, the Kraus index doesn't need to be rotated, but doing so leads to a character-free expression for $B$.}, and also transforming the initial state to be $\tilde{\omega}=U^\dagger \omega U=\ketbra{1}{1}_{\widehat{G}}$. In components, $B^{\chi\tilde{a}}_{\tilde{\alpha}\tilde{\beta}}=\frac{1}{\abs{G}^2}\sum_{i,a,\alpha,\beta}A^{ia}_{\alpha\beta}\,\tilde{\alpha}(\alpha)\,\tilde{\beta}(\beta)\,\chi(i)\,\tilde{a}(a)$, which after simplifying yields $B^{\chi\tilde{a}}_{\tilde{\alpha}\tilde{\beta}}=\frac{1}{\sqrt{\abs{G}}}\delta^{\chi,\tilde{a}}\delta_{\tilde{\alpha},\tilde{a}^{-1}\tilde{\beta}}$. The associated channel $\mathcal{E}_{B}(\ketbra{x}{y})=\frac{1}{\abs{G}}\sum_{\chi\in\widehat{G}} \ketbra{\chi^{-1}x}{\chi^{-1}y}\otimes\ketbra{\chi}{\chi}$ updates the memory with the inverse of the emitted charge at each step which guarantees global neutrality in the combined physical and memory system. Thus we have $\rho^{(N)}(B;\ketbra{1}{1}_{\widehat{G}})=\tilde{\rho}^{(N+1)}_G$.

\section{Details on sLPDO with weak $\mathbb{Z}_2$ symmetry}
\label{app:Z_2_sLPDO}
We provide details on the calculation to show that the state introduced in section~\ref{sec:Z_2_sLPDO} satisfies the intertwiner relation Eq.~\eqref{eq:MPU_intertwiner}.

Let $A^{\{2\}}$ be as in Eq.~\eqref{eq:Z_2_sLPDO_purification_tensor}. We denote $\sigma^{(M)}(A^{\{2\}},\sqrt{\omega})$ as the purification on $M+1$ sites with physical dimension $d=4$. That means there are $2M+2$ sites with physical dimension $d=2$. Denoting $N=2M$, in the following we consider the purification as a map from the ancilla space $\ket{(a_N,a_{N-1}),\dots,(a_2,a_1),(a_0,c)}$ to physical space $\ket{(a_N,c),(a_{N-1},\dots,a_2),(a_1,a_0)}$, which we simply abbreviate as $\ket{\vb{a},c}$ and $\ket{\pi(\vb{a},c)}$, respectively. We encourage the reader to sketch out the corresponding tensor network diagram, from which these statements are obvious.

When multiplying the tensors to form the purification $\sigma^{(M)}(A^{\{2\}},\sqrt{\omega})$, the internal Hadamards and $CZ$'s cancel, leaving only a single $CZ_{c,a_N}$. Then, writing $\sigma^{(M)}(A^{\{2\}},\sqrt{\omega})=H^{\otimes(N+2)}\eta^{(M)}(A^{\{2\}},\sqrt{\omega})$, we have
\begin{equation}
    \eta^{(M)}(A^{\{2\}},\sqrt{\omega})\ket{\vb{a},c}=R_{\vb{a},c}(-1)^{c (Q(\vb{a},c)+S(\vb{a},c))}\ket{\pi(\vb{a},c)}\,,
\end{equation}
where
\begin{equation}
    R_{\vb{a},c}=\frac{1}{\sqrt{2}}\lambda_{ca_0}\prod_{j=1}^N \lambda_{a_{j-1}a_j}\,,
\end{equation}
and phases 
\begin{align}
    Q(\vb{a},c)&=ca_0+\sum_{j=1}^N a_{j-1}a_j+ca_N\,, \\
    S(\vb{a},c)&=c+\sum_{j=0}^Na_j\,.
\end{align}
These follow from collecting all the coefficients and phases from the tensor network diagram.

The intuition for such a construction can be seen when we perform an on-site spin flip, that all non-trivial transformations come from the phases. The phases will be rewritten as an operator acting on the purification degrees of freedom, giving the intertwining relation Eq.~\eqref{eq:MPU_intertwiner}. We will now explicitly show this intertwining relation.

First, we act
\begin{align}
    U_{\mathrm{a}}^{(N+2)}\ket{\vb{a},c}&=(-1)^{(N+2)/2} CZ_{c,a_0}\prod_{j=1}^N CZ_{a_{j-1},a_j}CZ_{c,a_N} \nonumber \\
    &\phantom{={}}\times Y^{\otimes(N+2)}\ket{\vb{a},c} \nonumber \\
    &=(-1)^{c+\sum_{j=0}^N a_j+\overline{c}\,\overline{a_0}+\sum_{j=1}^N\overline{a_{j-1}}\,\overline{a_j}+\overline{c}\,\overline{a_N}}\ket{\overline{\vb{a}},\overline{c}} \nonumber \\
    &=(-1)^{Q(\overline{\vb{a}},\overline{c})+S(\vb{a},c)}\ket{\overline{\vb{a}},\overline{c}}\,,
\end{align}
where $\overline{a}=a+1$. The second equality follows from the actions of $Y$ and $CZ_{i,i+1}$ on the computational basis. Then
\begin{align}
    \eta^{(M)}U^{(N+2)}_a\ket{\vb{a},c}&=R_{\overline{\vb{a}},\overline{c}}(-1)^{Q(\overline{\vb{a}},\overline{c})+S(\vb{a},c)+\overline{c}(Q(\overline{\vb{a}},\overline{c})+S((\overline{\vb{a}},\overline{c})))} \nonumber\\
    &\phantom{={}}\times\ket{\pi(\overline{\vb{a}},\overline{c})}\,.
\end{align}
We now use the properties that $R_{\overline{\vb{a}},\overline{c}}=R_{\vb{a},c}$, $Q(\overline{\vb{a}},\overline{c})=Q(\vb{a},c)$ and $S(\overline{\vb{a}},\overline{c})=S(\vb{a},c)$ when $N$ is even, which is true since $N=2M$. Thus, one has that 
\begin{align}
    \eta^{(M)}U^{(N+2)}_a\ket{\vb{a},c}&=R_{\vb{a},c}(-1)^{c(Q(\vb{a},c)+S(\vb{a},c))}\ket{\pi(\overline{\vb{a}},\overline{c})} \nonumber \\
    &=R_{\vb{a},c}(-1)^{c(Q(\vb{a},c)+S(\vb{a},c))}X^{\otimes(N+2)}\ket{\pi(\vb{a},c)} \nonumber \\
    &=X^{\otimes (N+2)}\eta^{(M)}\ket{\vb{a},c}\,.
\end{align}
Therefore, we establish the operator equality $\eta^{(M)}U^{(N+2)}_a=X^{\otimes (N+2)}\eta^{(M)}$ and recalling that $\sigma^{(M)}=H^{\otimes(N+2)}\eta^{(M)}$, it follows that 
\begin{equation}
    Z^{\otimes(N+2)}\sigma^{(M)}=(-1)^{M+1}\sigma^{(M)}CZY^{(N+2)}\,.
\end{equation}

Before we conclude, we derive the expression for the density operator~\eqref{eq:Z_2_sLPDO_state}. We have
\begin{equation}
    \eta^{(M)}(\eta^{(M)})^\dagger = \sum_{\vb{a},c} R^2_{\vb{a},c} \ketbra{\pi(\vb{a},c)}{\pi(\vb{a},c)}\,.
\end{equation}
We note that $\lambda_{xa}^2=\frac{1+\mu(-1)^{x+a}}{2}$, which gives the precisely the Ising product $(\id+Z_{i_{k-1}}Z_{i_k})$ for each adjacent physical index. Conjugating by on-site Hadamards gives the Ising product in the rotated $X$-basis.

\subsection{MPO dressing}
\label{app:Z_2_sLPDO_MPO}
In this appendix, we construct an invertible MPO which commutes with the physical symmetry $U_{\mathrm{p}}^{(N+2)}=Z^{\otimes(N+2)}$. For $k=1,\dots,N$, we define the generators 
\begin{equation}
    G_k=Z_{i_{k-1}}X_cX_{i_{k}}\,, \qquad \text{satisfying } [Z^{\otimes{N+2}},G_k]=0\,.
\end{equation}
Furthermore, the generators have the properties $G_k^2=\id$, $\{G_k,G_{k+1}\}=0$ and $[G_k,G_{k'}]=0$ for $\abs{k-k'}>1$. We introduce the MPO operator $\mathcal{M}_\alpha^{(N)}$ as a product of non-unitary operators given by
\begin{equation}
\label{eq:MPO_dressing}
    \mathcal{M}^{(N)}_\alpha=O_N(\alpha)O_{N-1}(\alpha)\cdots O_1(\alpha)\,,
\end{equation}
where
\begin{equation}
    O_k(\alpha)=\frac{e^{\alpha G_k}}{\sqrt{\cosh(2\alpha)}}\,,\qquad \alpha\in\mathbb{R}
\end{equation}
normalized appropriately as to preserve the norm of the bare state $\rho^{(M)}$ when $\mathcal{M}_\alpha^{(N)}$ acts on it. $\mathcal{M}_{\alpha}^{(N)}$ can be realized as a bond dimension $4$ MPO because it can be written as a staircase of nearest-neighbor operators.

The dressed state is given by
\begin{equation}
    \sigma^{(M)}_{\alpha}=\mathcal{M}^{(N)}_\alpha \sigma^{(M)}\,, \qquad \rho^{(M)}_{\alpha}=\mathcal{M}^{(N)}_\alpha \rho^{(M)}(\mathcal{M}^{(N)}_\alpha)^\dagger\,,
\end{equation}
which exhibits vastly different properties from the bare state $\rho^{(M)}$ which we will comment on later. Because $\mathcal{M}_\alpha^{(N)}$ commutes with $Z^{\otimes (N+2)}$, the purification $\sigma_\alpha^{(M)}$ continues to satisfy the intertwining relation~\eqref{eq:MPU_intertwiner} and the state $\rho_\alpha^{(M)}$ still possesses the weak $\mathbb{Z}_2$ symmetry.

The global MPO $\mathcal{M}^{(N)}_\alpha$ can be absorbed into a new purification tensor $\tilde{A}$, such that the pair $(\tilde{A}^{\{2\}},\sqrt{\omega})$ generates $\sigma^{(M)}_\alpha$. Locally, $\tilde{A}$ is obtained from $A$ by acting on the emitted physical index $i$ and outgoing virtual indices $(c',y)$ with the operator
\begin{equation}
    o_\alpha=\frac{e^{\alpha Z_i X_c'X_y}}{\sqrt{\cosh(2\alpha)}}\,, \qquad \tilde{A}=o_\alpha A\,.
\end{equation}
Because of the way $A$ is constructed, at the $k^{\mathrm{th}}$ sequential step, the emitted physical index carries $i_{k-1}=a_{k-1}$, while the outgoing virtual indices carry $(c,a_k)=(c,i_k)$. Thus each tensor precisely contributes a factor of $O_k(\alpha)$ to the final state. Since the action of $o_\alpha$ was only on the outgoing physical and virtual degrees of freedom, and it is invertible, the tensor $\tilde{A}^{\{2\}}$ remains step-injective.

Finally, we show that the dressed state $\rho^{(M)}_\alpha$ contains non-trivial long-range correlations, which sets it apart from the bare state $\rho^{(M)}$. The intuition comes from that each factor appearing in Eq.~\eqref{eq:MPO_dressing} contains an $X_c$ operator. Define  $Q_k=Z_{i_{k-1}}X_{i_k}$, which is a local observable in the qubit chain. It is odd under the $\mathbb{Z}_2$ symmetry $Z^{\otimes(N+2)}Q_kZ^{\otimes(N+2)}=-Q_k$, so its expectation value with the state vanishes $\expval{Q_k}_{\rho^{(M)}_\alpha}=0$. However, its two-point correlator can be calculated using the algebraic properties of the operators to be 
\begin{equation}
    \expval{Q_k Q_{k'}}_{\rho^{(M)}_\alpha}=\left(\frac{\tanh(2\alpha)}{\cosh(2\alpha)}\right)^2
\end{equation}
for non-adjacent sites $1\leq k<k'-1$, $k'\leq N-1$. This shows that long-range order persists in the state $\rho^{(M)}_\alpha$, which is not present in the bare state $\rho^{(M)}$.

\bibliography{references}

\end{document}